\documentclass[twocolumn]{article}
\usepackage{mglgrid}

\usepackage{scalerel}[2016/12/29]
\usepackage{stackengine}
\usepackage{nicefrac}
\usepackage[normalem]{ulem}
\usepackage{siunitx}
\usepackage{mathtools}
\usepackage{amsmath}
\usepackage{amsthm}
\usepackage{thmtools}
\usepackage{thm-restate}
\usepackage{balance}

\newcommand*\lisp[1]{\texttt{#1}}

\allowdisplaybreaks
\let\leq\leqslant

\renewcommand{\neg}{\protect\hstretch{0.8}{-}}

\DeclarePairedDelimiter{\floor}{\lfloor}{\rfloor}
\DeclarePairedDelimiter{\abs}{\lvert}{\rvert}

\newcommand*{\natnum}{\mathbb{N}}
\newcommand*{\natnumzero}{\mathbb{N}_0}
\newcommand{\E}{\mathop{{}\mathbb{E}}}
\newtheorem{theorem}{Theorem}
\newtheorem{definition}[theorem]{Definition}

\newcommand\nobreakpar{\par\nobreak\@afterheading}
\makeatother

\newcommand\pig[1]{\scalerel*[5.5pt]{\Big#1}{%
  \ensurestackMath{\addstackgap[1.55pt]{\big#1}}}}
\newcommand\pigl[1]{\mathopen{\pig{#1}}}
\newcommand\pigr[1]{\mathclose{\pig{#1}}}

\newcommand\Pig[1]{\scalerel*[5.5pt]{\bigg#1}{%
  \ensurestackMath{\addstackgap[1.55pt]{\Big#1}}}}
\newcommand\Pigl[1]{\mathopen{\Pig{#1}}}
\newcommand\Pigr[1]{\mathclose{\Pig{#1}}}

\newcommand\pigg[1]{\scalerel*[5.5pt]{\Bigg#1}{%
  \ensurestackMath{\addstackgap[1.55pt]{\bigg#1}}}}
\newcommand\piggl[1]{\mathopen{\pigg{#1}}}
\newcommand\piggr[1]{\mathclose{\pigg{#1}}}

\usepackage{color}
\usepackage{xcolor}
\definecolor{mglred}{RGB}{165,30,30}
\definecolor{mglgreen}{RGB}{37,125,49}
\definecolor{mglblue}{RGB}{65,48,163}

\usepackage{booktabs}
\usepackage{multirow}
\usepackage{algorithm}

\usepackage{array}
\newcolumntype{L}[1]{>{\raggedright\let\newline\\\arraybackslash
    \hspace{0pt}}m{#1}}
\newcolumntype{C}[1]{>{\centering\let\newline\\\arraybackslash
    \hspace{0pt}}m{#1}}
\newcolumntype{R}[1]{>{\raggedleft\let\newline\\\arraybackslash
    \hspace{0pt}}m{#1}}

\usepackage{caption}
\usepackage{tikz}
\usepackage{pgfplots,tikz,pgfplotstable}
\pgfplotsset{compat=1.9}
\usetikzlibrary{positioning, fit, arrows.meta, shapes}
\newcommand\tikzpicturedependsonfile[1]{}

\pgfplotsset{
    legend image with text/.style={
        legend image code/.code={%
            \node[anchor=center] at (0.3cm,0cm) {#1};
        }
    },
}
\pgfplotsset{uniformhash/.style = {black!50!white, loosely dotted, very thick}}
\pgfplotsset{murmurhash/.style = {teal!70!black, densely dotted, very thick}}
\pgfplotsset{prefuzzhash/.style = {blue!70!black, densely dashed, thick}}
\pgfplotsset{constanthash/.style = {orange!70!white, densely dashdotted, thick}}
\pgfplotsset{adaptivehash/.style = {red!70!black, semithick}}

\pgfplotsset{%
  redweak/.style = {red!70!black,densely dotted,
                    mark=o, mark options={scale=0.7,solid}}}
\pgfplotsset{%
  redhollow/.style = {red!100!white,densely dashed,
                      mark=o, mark options={scale=1.0,solid}}}
\pgfplotsset{%
  redsolid/.style = {red!50!white,mark=*,mark options={scale=1.0,solid}}}

\pgfplotsset{%
  orangeweak/.style = {orange!70!black,densely dotted,mark=diamond,
                       mark options={scale=1.2*0.7,solid}}}
\pgfplotsset{%
  orangehollow/.style = {orange!100!white,densely dashed,mark=diamond,
                         mark options={scale=1.2,solid}}}
\pgfplotsset{%
  orangesolid/.style = {orange!50!white,mark=diamond*,
                        mark options={scale=1.2,solid}}}

\pgfplotsset{%
  violetweak/.style = {violet!70!black,densely dotted,mark=diamond,
                       mark options={scale=1.2*0.7,solid}}}
\pgfplotsset{%
  violethollow/.style = {violet!100!white,densely dashed,mark=diamond,
                         mark options={scale=1.2,solid}}}
\pgfplotsset{%
  violetsolid/.style = {violet!50!white,mark=diamond*,
                        mark options={scale=1.2,solid}}}

\pgfplotsset{%
  blueweak/.style = {blue!70!black,densely dotted,mark=square,
                     mark options={scale=0.8*0.7,solid}}}
\pgfplotsset{%
  bluehollow/.style = {blue!100!white,densely dashed,mark=square,
                       mark options={scale=0.8,solid}}}
\pgfplotsset{%
  bluesolid/.style = {blue!50!white,mark=square*,
                      mark options={scale=0.8,solid}}}

\pgfplotsset{%
  tealweak/.style = {teal!70!black,densely dotted,mark=triangle,
                     mark options={scale=1.2*0.7,solid}}}
\pgfplotsset{%
  tealhollow/.style = {teal!100!white,densely dashed,mark=triangle,
                       mark options={scale=1.2,solid}}}
\pgfplotsset{%
  tealsolid/.style = {teal!50!white,mark=triangle*,
                      mark options={scale=1.2,solid}}}

\pgfplotscreateplotcyclelist{acycle}{%
solid, every mark/.append style={solid, fill=gray}, mark=*\\%
dotted, every mark/.append style={solid, fill=gray}, mark=square*\\%
densely dotted, every mark/.append style={solid, fill=gray}, mark=otimes*\\%
loosely dotted, every mark/.append style={solid, fill=gray}, mark=triangle*\\%
dashed, every mark/.append style={solid, fill=gray},mark=diamond*\\%
loosely dashed, every mark/.append style={solid, fill=gray},mark=*\\%
densely dashed, every mark/.append style={solid, fill=gray},mark=square*\\%
dashdotted, every mark/.append style={solid, fill=gray},mark=otimes*\\%
dashdotdotted, every mark/.append style={solid},mark=star\\%
densely dashdotted,every mark/.append style={solid, fill=gray},mark=diamond*\\%
}

\usepackage[hidelinks,linktoc=all]{hyperref}
\usepackage{hypcap}
\usepackage{bookmark}
\usepackage[nameinlink]{cleveref}
\usepackage[noend]{algpseudocode}
\usepackage{algorithmicx}
\newcommand*\Let[2]{\State #1 $\gets$ #2}
\algdef{SE}[DOWHILE]{Do}{doWhile}{\algorithmicdo}[1]{\algorithmicwhile\ #1}%

\title{Adaptive Hashing}
\subtitle{Faster Hash Functions with Fewer Collisions*}

\author{\authorname{G\'abor Melis}\\
  \authoremail{melisgl@google}.\authoremail{com}\\
  \authororg{Google DeepMind}}

\begin{document}

\maketitle

\abstract{Hash tables are ubiquitous, and the choice of hash function, which maps a key to a bucket, is key to their performance.
We argue that the predominant approach of fixing the hash function for the lifetime of the hash table is suboptimal and propose adapting it to the current set of keys.
In the prevailing view, good hash functions spread the keys ``randomly'' and are fast to evaluate.
General-purpose ones (e.g. Murmur) are designed to do both while remaining agnostic to the distribution of the keys, which limits their bucketing ability and wastes computation.
When these shortcomings are recognized, one may specify a hash function more tailored to some assumed key distribution, but doing so almost always introduces an unbounded risk in case this assumption does not bear out in practice.
At the other, fully key-aware end of the spectrum, Perfect Hashing algorithms can discover hash functions to bucket a given set of keys optimally, but they are costly to run and require the keys to be known and fixed ahead of time.
Our main conceptual contribution is that adapting the hash table's hash function to the keys online is necessary for the best performance, as adaptivity allows for better bucketing of keys \emph{and} faster hash functions.
We instantiate the idea of online adaptation with minimal overhead and no change to the hash table API.
The experiments show that the adaptive approach marries the common-case performance of weak hash functions with the robustness of general-purpose ones.}


\section{Introduction}

{
\renewcommand*{\thefootnote}{*}
\begin{NoHyper}
\footnotetext{\dots Especially in Certain Situations}
\end{NoHyper}
}

Hash tables \citep{luhn1953new,enwiki:1207615479} map keys to values and are one of the most fundamental data structures.
As such, their performance is of considerable interest.
For example, \Citet{hentschel2022entropy} claimed in 2022 that ``a typical complex database query (TPC-H) could spend 50\% of its total cost in hash tables, while Google spends at least 2\% of its total computational cost across all systems on C++ hash tables''.
At that scale, even small gains in this area can have huge impact.
This work aims to provide a general framework to improve the performance of hash tables in practice.
Theory mostly concerns itself with the unknown key distribution setting and the cost of key lookup abstracted as the number of key comparisons required.
This approach is highly successful due to its wide applicability and for being a reasonable model \emph{asymptotically} for a certain class of hash functions.
However, practice stubbornly happens in the non-asymptotic regime with particular keys on computers with memory caches.
We propose a method to equip hash tables with the ability to change their hash function based on the keys being added to improve their overall performance by bucketing the keys more evenly, making the hash function faster or more cache-friendly.
We are not aware of a cost model that incorporates these effects and is amenable to theoretical analysis, hence this work tilts heavily towards Experimental Algorithmics \citep{moret1999towards}.

\begin{algorithm}[t]
  \caption{\small Sketch of a possible adaptation mechanism implemented in $\mathit{put}$ with low overhead.
The unchanged parts of a simple $\mathit{put}$ implementation are \textcolor{gray!100}{grayed out}, while extension points are marked in \textcolor{mglred}{red}.
The hash function \emph{hash\_fn} may be adapted when there are too many keys in the same bucket, or when rehashing finds that the number of collisions is too high.
Alternatively, the total cost of access (\Cref{def:cost-of-hashes}) could be tracked, requiring an additional write to memory and also a performance penalty to key deletion.}
  \label{alg:puthash}
  \hrule
  \begin{algorithmic}[1]
    \color{gray!100}
    \Function{$\operatorname{put}$}{$\mathit{key}$, $\mathit{value}$}
      \Let{$\mathit{bucket}$}{$\operatorname{hash\_fn}(\mathit{key}) \bmod m$}
      \color{black}
      \Let{$\mathit{chain\_length}$}{$0$}
      \color{gray!100}
      \For{$k$ $\gets$ next key in $\mathit{bucket}$}
        \If{$\operatorname{compare}(\mathit{key}, k)$}
          \Let{value of $k$}{$\mathit{value}$}
          \State\Return
        \EndIf
        \color{black}
        \Let{$\mathit{chain\_length}$}{$\mathit{chain\_length} + 1$}
        \color{gray!100}
      \EndFor
      \color{black}
      \If{$\mathit{chain\_length}$ \textcolor{mglred}{too high}}
        \Let{$\operatorname{hash\_fn}$}{$\textcolor{mglred}{\operatorname{safer\_hash\_function}}(\operatorname{hash\_fn})$}
        \State{$\operatorname{rehash\textcolor{mglred}{\_and\_maybe\_adapt}}()$}
        \Let{$\mathit{bucket}$}{$\operatorname{hash\_fn}(\mathit{key}) \bmod m$}
      \EndIf
      \color{gray!100}
      \If{hash table is full}
        \State{double $m$ and increase storage}
        \color{black}
        \State{$\textcolor{gray}{%
            \operatorname{rehash\textcolor{mglred}{\_and\_maybe\_adapt}}()}$}
        \If{$\operatorname{hash\_fn}$ was changed}
          \Let{$\mathit{bucket}$}{$\operatorname{hash\_fn}(\mathit{key}) \bmod m$}
        \EndIf
        \color{gray!100}
      \EndIf
      \State{add $(\mathit{key}, \mathit{value})$ to $\mathit{bucket}$}
    \EndFunction
  \end{algorithmic}
  \hrule
\end{algorithm}

It is useful to contrast our method with hand-crafting hash functions for particular key distributions.
We argue that hand-crafting is too rigid: it breaks badly if the distributional assumption is violated \citep{jdkhashbug,jdkstringhashsubsampling}.
It is also costly in terms of human labour, and designing hash functions with guarantees is hard.
While selecting the hash function offline amortizes the design cost, it also attempts to solve a much harder problem than necessary in balancing the complexity and the quality of the hash function across all possible sets of keys.

A less rigid solution is to select a hash function \emph{online} given the actual keys in the hash table.
This pulls the hash function selection cost into the runtime realm, and we must be extremely cautious of introducing overhead lest any possible gains be wasted.
Here, we propose hiding some of the selection cost in rehashing, when the hash table grows.
To ground the exposition, \Cref{alg:puthash} sketches the implementation of a possible adaptation mechanism.
In later sections, we flesh out this skeletal algorithm by instantiating the logic behind the extension points marked in red for deciding when the chain length is considered too high, the fallback mechanism to a safer hash, and the adaptation in rehash.

In summary, our main contributions are:
\begin{itemize}
\item We argue that hash functions must adapt to the keys in the hash table for best performance.
In \Cref{sec:the-case-for-adaptive-hashing}, we quantify the theoretical cost ignoring the actual keys.
\item We propose a light-weight adaptation mechanism tailored to string hashing and achieve large gains by ensuring that the hash function performs only the minimum work necessary in \Cref{sec:half-an-example}.
\item We propose a robust solution for integer/pointer keys that simultaneously achieves faster hashing and fewer collisions in \Cref{sec:pointer-keys}.
\end{itemize}
The proposed algorithms were implemented within SBCL \citep{newman1999sbcl}, a high-performance Common Lisp.
SBCL hash tables are reasonably fast given the constraints of the ANSI standard \citep{steele1990common} but not state of the art.
We microoptimized the baseline SBCL v2.4.2 to make performance comparisons fairer.


\section{The Case for Adaptive Hashing}
\label{sec:the-case-for-adaptive-hashing}

In this section, we present the traditional cost model in hash tables, based on the number of comparisons, then characterize how much being key-agnostic costs in these terms.
We denote the number of buckets with $m$, the number of keys with $n$, the keys (assumed to be integers) with $k$, hash values with $h \in \natnumzero$, buckets with $b \in [0,m-1]$, and say that hash $h$ falls into bucket $b$ if $h \bmod m = b$.
We write $h_{1:n}$ to refer to a vector of a certain size.

\begin{definition}[Bucket Count]
For a given set of hash values $h_{1:n}$ and $m$ buckets, we denote the bucket count vector with $c(h_{1:n}, m) \in (\natnumzero)^m$, where $c(h_{1:n}, m)_b = \vert\{i \colon i \in [1, n], h_i \bmod m = b\}\vert$ is the number of hashes falling into bucket $b$, for all $b \in [0,m-1]$.
\end{definition}

Next, we define the cost of hash values at a given number of buckets to be the expected number of comparisons one has to make to find the value associated with a key present.
For a single bucket with $c_b$ hashes, this is $(c_b+1)/2$ assuming a uniform distribution over the keys being looked up.

\begin{definition}[Cost of Hashes]
\label{def:cost-of-hashes}
The cost of hashes $h_{1:n}$ with $m$ buckets is $C(c) = n^{-1} \sum_{b=0}^{m-1} c_b(c_b+1)/2$, using the shorthand $c=c(h_{1:n}, m)$.
\end{definition}
\noindent Note that this definition differs from \emph{hash function quality} in the Dragon Book \citep{dragonbook} only in the normalization.

We define perfect hashes as those that fill all buckets as equally as possible.
This generalizes the classic definition \citep{fredman1984storing}, which requires no collisions, to the space-restricted setting of $m<n$.
\begin{definition}[Perfect Hash]
We say the hashes $h_{1:n}$ are perfectly distributed in $m$ buckets if $m - (n \bmod m)$ buckets have $\floor{n/m}$ hashes in them, and $n \bmod m$ buckets have $\floor{n/m} + 1$ hashes in them.
\end{definition}

\vspace{-\groundskip}
\begin{restatable}{proposition}{propperfecthasheshaveminimalcost}[Minimal Cost]
Let $U(n,m)$ be the bucket count vector of any perfect hash of $n$ keys and $m$ buckets.
Let $q=\floor{n/m}$ and $r=n \bmod m$.
Then,
\begin{align*}
C(U(n,m)) = (m-r)\frac{q(q+1)}{2n} + r\frac{(q+1)(q+2)}{2n},
\end{align*}
and this cost is minimal.
\end{restatable}
\noindent For the proof of this and other propositions, see \Cref{sec:proofs}.
In the illustrative special case of an integer load factor $n/m$, we have that $n=qm$ (i.e. $r = 0$), the counts will be the same for all buckets, and $C(U(n,m)) = m\frac{q(q+1)}{2n} = \frac{q+1}{2}$.

Since perfect hashes have minimal cost, we define the regret of a hash the excess cost over that.
\begin{definition}[Regret of Hashes]
\label{def:regret}
The regret of hashes $h_{1:n}$ is
\begin{align*}
R(h_{1:n},m) = C(c(h_{1:n},m))- C(U(n,m)).
\end{align*}
\end{definition}

Note that the classic cost model in \Cref{def:cost-of-hashes} is simple but clearly wrong if hashes of keys are cached (see \Cref{alg:puthash}) because it costs one memory access to look up the cached hash for a key, but the cost of key comparison may be much higher.
Still, with random hashes of many bits, this distinction becomes moot for regret because only one cached hash is likely to match, so there will be a single comparison for each lookup, and their contributions to $C(c(h_{1:n},m))$ and $C(U(n,m))$ cancel out.
Thus, with cached hashes, the regret can be interpreted to be in terms of memory access.

Hash functions strive to mimic the uniform hash \citep{10.5555/1213024}, which assigns keys to buckets with equal probability.
We consider the cost of this ideal next.

\begin{restatable}{proposition}{propexpectedcostofuh}[Expected Cost of the Uniform Hash]
\label{prop:expected-cost-of-uh}
Let $P$ be a uniform distribution over functions that map keys to buckets.
Then,
\begin{align*}
\E_{\pi_{1:n} \sim P} C(c([\pi_1(k_1), \dots \pi_n(k_n)],m)) = 1+\frac{n-1}{2m},
\end{align*}
where $\pi_{1:n}$ are $n$ independent samples from $P$.
\end{restatable}

The uniform hash is optimal among hashes that are functions of a single key.
However, its cost is clearly worse when compared to a perfect hash, which can be viewed as having knowledge of all keys.
Next, we characterize its regret, assuming that $m \mid n$, for convenience.

\begin{restatable}{proposition}{propexpectedregretofuh}[Expected Regret of the Uniform Hash]
\label{prop:regret-of-uh}
For all load factors $q \in \natnum$ ($n=qm$), the expected regret of the uniform hash is $0.5 + \frac{1}{m}$.
\end{restatable}

So, the uniform hash wastes $0.5$ comparisons per lookup compared to a perfect hash because only considers a single actual key.
This may be worth improving, especially if the hash function can be made faster at the same time.


\section{Half an Example}
\label{sec:half-an-example}

\begin{algorithm}[t]
  \caption{\small Hashing strings with a $\mathit{limit}$ on the number of characters taken into account.
The algorithm moves inwards from the two ends of the string because those tend to be the most informative and because this scheme can be easily extended to reuse a previously computed hash with a lower limit.
The function $\mathit{add\_char}$ performs one step of the FNV-1A algorithm.}
  \label{alg:string-hash}
  \hrule
  \begin{algorithmic}[1]
    \Function{$\operatorname{hash\_string}$}{$s$, $\mathit{limit}$}
      \Let{$h$}{$\operatorname{len}(s)$}
      \Comment{Initialize the hash to the length}
      \Let{$a, b$}{$0, \operatorname{len}(s)-1$}
      \Let{$n$}{$\min(\mathit{limit}, h)$}
      \While{$a < (n \gg 1)$}
        \State{$h \gets \operatorname{add\_char}(h, s[a])$}
        \State{$h \gets \operatorname{add\_char}(h, s[b])$}
        \State{$a, b \gets a+1, b-1$}
      \EndWhile
      \If{$n \bmod 2 = 1$}
        \Comment{Add the odd middle char}
        \State{$h \gets \operatorname{add\_char}(h, s[a])$}
      \EndIf
      \State\Return{$h$}
    \EndFunction
  \end{algorithmic}
  \hrule
\end{algorithm}

In related work, \citet{hentschel2022entropy} engage with one half of this problem: they select high-entropy parts of the key to feed to a general-purpose hash function.
Their approach can speed up the hash function but cannot reduce the expected number of collisions.
Crucially, once a hash function has been learned in an offline manner for a given key distribution, it remains fixed for the lifetime of the hash table.
In a similar vein but adapting the hash function on the fly, we demonstrate significant speedups on string hashing even with slightly more collisions.

In particular, we hash only a subset of the data in compound keys, where the size of the subset is subject to a dynamically adjusted limit.
In case of string keys, we limit the number of characters hashed.
Hashing proceeds inwards alternating between taking a character from the beginning and the end of the string.
The algorithm (FNV-1A \citep{enwiki:1215467221}) is initialized with the length of the string to cheaply introduce some information about the truncated away characters into the hash.
See \Cref{alg:string-hash} for the code listing.

As a heuristic to detect overly severe truncation, we track the maximum chain length.
That is, when a new key is being inserted whose hash is computed with truncation (e.g. it's a string longer than the current limit), we check the number of keys already in its bucket.
If the probability of that many keys having collided with the uniform hash (without truncation) is less than 1\%, we double the limit and rehash.
Since hashes of strings are expensive to compute, they are cached (see \Cref{alg:puthash,sec:the-case-for-adaptive-hashing,sec:sbcl-hash-tables}).
By compromising the hash function's quality, we run the risk of having to perform more comparisons, which can be costly, thus, it is important to have a tight limit on the chain length, which we achieve by precomputing them for all possible power-of-2 bucket counts at load factor 1 and changing the current limit when the hash table is resized.

\newcommand\figstringexistingregret{
\begin{tikzpicture}
  \begin{axis}[xlabel=number of keys, ylabel=regret,
      xmin=1, xmode=log, log basis x={2},
      ymin=-0.01,
      legend pos=south east,
      legend style={nodes={scale=0.7, transform shape}},
      legend cell align={left},
      height=0.6*\linewidth,
      width=0.98\linewidth,
    ]
    \pgfplotstableread{data/string-existing-sbcl.tbl}{\sorted}
    \addplot[uniformhash] table [x=nkeys, y=rndregret] {\sorted};
    \addlegendentry{Uniform};

    \pgfplotstableread{data/string-existing-sbcl.tbl}{\sorted}
    \addplot[prefuzzhash] table [x=nkeys, y=regret] {\sorted};
    \addlegendentry{SBCL};

    \pgfplotstableread{data/string-existing-adaptive.tbl}{\sorted}
    \addplot[adaptivehash] table [x=nkeys, y=regret] {\sorted};
    \addlegendentry{\textbf{Adaptive}};
  \end{axis}
\end{tikzpicture}
\caption{Regret (\Cref{def:regret}) with string keys.
Adaptive does not gain or significantly compromise on regret.
Points where the truncation limit changes vary between runs.}
}

\newcommand\figstringexistingput{
\begin{tikzpicture}
  \begin{axis}[ylabel=ns / put,
      xmin=1, xmode=log, log basis x={2},
      ymin=50, ymax=200, ymode=log, log basis y={2},
      legend pos=north west,
      legend style={nodes={scale=0.7, transform shape}},
      legend cell align={left},
      height=0.6*\linewidth,
      width=0.98\linewidth,
    ]
    \pgfplotstableread{data/string-existing-sbcl.tbl}{\sorted}
    \addplot[prefuzzhash] table [x=nkeys, y=putns] {\sorted};
    \addlegendentry{SBCL};

    \pgfplotstableread{data/string-existing-adaptive.tbl}{\sorted}
    \addplot[adaptivehash] table [x=nkeys, y=putns] {\sorted};
    \addlegendentry{\textbf{Adaptive}};
  \end{axis}
\end{tikzpicture}
\caption{PUT timings in nanoseconds with string keys.
Note the log scales. The plot shows the \emph{average} time for inserting a new key when populating an empty hash table with a given number of keys.}
}

\newcommand\figstringexistingget{
\begin{tikzpicture}
  \begin{semilogxaxis}[ylabel=ns / get,
      xmin=1, xmode=log, log basis x={2},
      ymin=50, ymax=200, ymode=log, log basis y={2},
      legend pos=north west,
      legend style={nodes={scale=0.7, transform shape}},
      legend cell align={left},
      height=0.6*\linewidth,
      width=0.98\linewidth,
    ]
    \pgfplotstableread{data/string-existing-sbcl.tbl}{\sorted}
    \addplot[prefuzzhash] table [x=nkeys, y=getns] {\sorted};
    \addlegendentry{SBCL};

    \pgfplotstableread{data/string-existing-adaptive.tbl}{\sorted}
    \addplot[adaptivehash] table [x=nkeys, y=getns] {\sorted};
    \addlegendentry{\textbf{Adaptive}};
  \end{semilogxaxis}
\end{tikzpicture}
\caption{GET timings with string keys.}
}

\begin{figure}[t]
\figstringexistingregret
\label{fig:string-existing-regret}
\end{figure}

If there are significantly more collisions after this rehash than would be expected with the uniform hash, then we double the limit again and rehash.
This procedure repeats until the number of collisions falls near the expected level (see \Cref{sec:number-of-collisions} for the details of our approximation) or there are no more keys with truncated hashes\footnote{We use the highest bit in the hash to indicate truncation.}.

\subsection{Experiments with String Keys}
\label{sec:experiments-with-string-keys}

We implemented adaptive hashing in SBCL's standard \lisp{equal} hash tables\footnote{Common Lisp's \lisp{equal} is like Java's \texttt{.equals()}: it compares two objects by value.}.
Then, we collected all different strings present in the running Lisp, which gave us about \num{40000} keys and measured the time it takes to populate hash tables from an empty state, averaging over same-sized random subsets of the keys.
We partitioned the range of possible key counts into maximally large segments within which the hash table internal data structures are not resized and measured performance with the lowest and highest possible key count in each segment.
We also measured the time it takes to look up an existing key (GET), to look up a key not in the hash table (MISS), and to delete an existing key (DEL).
All reported times are in nanoseconds per operation (e.g. insertions for populating the table, lookups for GET).
For details of the experimental setup, see \Cref{sec:microbenchmarking-methodology}.

\begin{figure}[t]
\figstringexistingput
\label{fig:string-existing-put}
\end{figure}

\begin{figure}[t]
\figstringexistingget
\label{fig:string-existing-get}
\end{figure}

\Cref{fig:string-existing-regret,fig:string-existing-put,fig:string-existing-get} show our results.
The jaggedness of the lines is the effect of hash table resizing.
At smaller sizes, the gains are considerable and similar to those in \citet{hentschel2022entropy}.
With the current implementation, the performance of the adaptive method eventually falls back to the baseline (unmodified SBCL) because the max-chain-length check in \Cref{alg:puthash} gets triggered by strings that share a long common prefix and suffix, pushing the truncation limit beyond the length of most keys.

Note that a more advanced implementation could reduce the overhead of rehashing by starting \Cref{alg:string-hash} from the hash value produced with a lower $limit$ and only hashing the characters beyond that.
While this would help PUT results a bit, GET would not benefit.
The more fundamental problem is that max-chain-length is a really loose indicator of the cost, and we should track the regret instead.

\subsection{Experiments with List Keys}
\label{sec:experiments-with-list-keys}

With a small modification, the solution used in the string case also works for lists, which is another common type of key in \lisp{equal} hash tables.
Since lists are not random-access, we only consider their prefixes (unlike \Cref{alg:string-hash}, which also includes suffixes of strings).
At least since the year 2000, stock SBCL has truncated list keys to length 4 to avoid stack exhaustion in case its recursive hashing algorithm is invoked on a circular key.
This value might have been chosen empirically to maximize performance, or user code has adapted to this limitation even if the root cause of the issue remained unrecognized, or both.
Regardless, we found that a default limit of 4 worked best.
From this default, the limit is increased as collisions warrant, similar to the string case.
So, the practical gains for the adaptive method may be limited to cases where crucial bits in keys are put unknowingly beyond length 4.
Curiously, there are two such cases in SBCL's own test suite (\texttt{\small arith-combinations.pure.lisp} and \texttt{\small save4.test.sh}), which experience a disastrous number of collisions and are sped up by 60\% when the adaptation mechanism increases the truncation length.

\section{Integer and Pointer Keys}
\label{sec:pointer-keys}

\Cref{sec:the-case-for-adaptive-hashing} indicates that it may be possible to improve the regret by adapting the hash function to the keys on the fly, but whether there is a practical implementation of adaptation -- which covers a non-trivial set of workloads and is lightweight enough to benefit overall performance -- remains to be demonstrated.

In the previous section, we showed an example of how to reduce the complexity of the hash without increasing the number of collisions too much.
In this section, we instantiate the general idea of adaptive hashing on the problem of hashing integer and pointer keys \citep{integerhashing}.
In this case, we will consider reducing the number of collisions at the same time as speeding up the hash function, which requires leveraging the underlying key distribution.

\subsection{Perfect Hashing on Arithmetic Sequences}
\label{sec:arithmetic-sequences}

First, we consider the idealized case of adding keys to a power-of-2 hash table from an arithmetic sequence of integers in order.
Let $a_i = a_0 + id$ for all $i \in [1, \dots]$, where $a_0$ is the offset and $d$ is difference between successive elements.
A perfect hash here is $h_i = \floor{a_i/d} = h_0 + i$, but this requires division by an arbitrary constant $d$, which is slow on current hardware.

Since the number of buckets $m$ is a power of 2, as long as $d$ is odd, any finite progression in $a$ will be perfectly distributed modulo $m$ because $d$ and $m$ are coprime.
Let $s$ be the largest integer such that $2^s \mid d$.
Then $h_i = \floor{a_i / 2^s}$ is an arithmetic sequence with odd increment $d / 2^s$, thus perfectly distributed.
So, if we know $s$, we can use the perfect hash function $k \rightarrow k \gg s$ with a single arithmetic shift.
Less regular than arithmetic sequences are the addresses of sequentially allocated objects, which we consider next.

\subsection{Page-Based Memory Allocators}
\label{sec:allocators}

If keys are memory addresses (e.g. pointers to objects), then we may be able to take advantage of how the allocator works.
We consider the case of page-based allocators, which first allocate contiguous memory ranges called pages from the OS.
From these pages, they are then able to allocate objects much more quickly.
To decrease contention, pages are often assigned to individual threads.
A thread may have multiple pages assigned to it, in which case it may choose between pages based on the allocation size.
In particular, TCMalloc allocates pages of 8KB by default and has allocations of roughly the same size within the same page.
SBCL, a Common Lisp implementation with a moving garbage collector (GC), has two 32KB pages per thread: one for conses (whose size is two machine words), and another for all other objects\footnote{This is a simplification. Some platforms also have immobile space, arenas, and ``large'' objects are handled specially.}.

Under such allocators, if the keys are of the same size and are allocated in a tight loop, we can expect their addresses to roughly follow an arithmetic progression.
But only roughly because with TCMalloc there may be holes (that belonged to previously freed objects) on the page, which may be filled in a more irregular pattern, and when a page is full, the new page may be anywhere in memory.
With SBCL, pages have no holes because allocation within a page is simply a pointer bump\footnote{The address of the next object is the address of the previously allocated object plus its size aligned to a double word boundary.} and because the GC compacts.
When there is not enough room left on the current page or GC happens, the allocator gets a new page, but the addresses of subsequent pages are much less regular than the addresses within pages.
A further complication with SBCL is that there are no separate pages for objects of different sizes, so if objects of non-constant sizes are allocated between subsequent keys, that throws regularity off and may reduce the density of addresses of keys within the page.

In summary, to the extent that addresses of keys are distributed following arithmetic progressions, their hashes can be improved.
We leverage the following properties:
\begin{itemize}
\item \textbf{denseness}: many keys are allocated on the same page,
\item \textbf{alignment}: the power-of-2 alignment of keys is constant (especially within a single page).
\end{itemize}

\subsection{Detecting Common Power-of-2 Factors}

We need to detect $s$ in the factor $2^s$ common to all keys fast and safely.
Fast because this will be done at runtime, and safely because using an overestimation of $s$ in e.g. the Arithmetic hash $k \rightarrow k \gg s$ can discard valuable bits and lead to a disastrous number of collisions, which must then be detected and corrected by another change to the hash function (see \Cref{alg:puthash} and \Cref{alg:rehash-eq}).

\Cref{alg:detect-shift} is designed to fulfil these requirements.
It is extremely light, performing about 2 bitwise assembly instructions per key, over only a subset of keys to limit memory access.
Also, it detects the common factor without assuming that the sequence is arithmetic, which makes it applicable in more circumstances.

\begin{algorithm}[t]
  \caption{\small Detecting common low bits in integer keys (e.g. pointers from page-based allocators) $k_1, \dots k_n$.
This is to find the largest power-of-2 factor of the common difference in an arithmetic progression regardless of the offset caused by the first term.
The symbols $\lor$, $\oplus$, $\lnot$ denote the bitwise OR, XOR and NOT operations.
Note that $\mathit{count\_leading\_zero\_bits}$ is often a single assembly instruction such as LZCNT on x86.}
  \label{alg:detect-shift}
  \hrule
  \begin{algorithmic}[1]
    \Function{$\operatorname{count\_common\_prefix\_bits}$}{$k_1, \dots, k_n$}
      \Let{$\mathit{mask}$}{$0$}
      \Comment{Changed bits detected so far.}
      \For{$i \gets 2$ to $n$}
        \State{$\mathit{mask} \gets \mathit{mask} \lor (k_1 \oplus k_i)$}
      \EndFor
      \State\Return{$\operatorname{count\_leading\_zero\_bits}(\lnot \mathit{mask})$}
    \EndFunction
  \end{algorithmic}
  \hrule
\end{algorithm}

As to safety, if we have $n$ keys, then the probability of a bit appearing constant by chance is $2^{1-n}$ (assuming that bits are $\text{Binomial}(0.5)$ in the hash values).
Thus, in practice, we can detect constant low bits with high probability with as few as 8-16 keys.
We use the detected shift $s$ in the following three hash functions.

\subsection{The Arithmetic Hash}
\label{sec:the-arithmetic-hash}

The Arithmetic hash function is $k \rightarrow k \gg s$, where $s \in \natnumzero$.
As discussed in \Cref{sec:arithmetic-sequences}, this is a perfect hash function for arithmetic progressions.

\subsection{The Pointer-Mix Hash}
\label{sec:the-pointer-mix-hash}

The Arithmetic hash function can easily have high cost if, for example, pointers on multiple pages come from the same smallish subset on each page.

Our next hash function, Pointer-Mix, combines the Arithmetic hash $k \gg s$ with a general purpose hash of the page address $k \gg \mathrm{PB}$, where $\mathrm{PB}$ is the base 2 logarithm of the allocation page size in bytes.
The Pointer-Mix hash function is $k \rightarrow k \gg s \oplus \operatorname{safe}(k \gg \mathrm{PB})$, where $\oplus$ is the bitwise XOR operation, and $\operatorname{safe}()$ is a general purpose hash function such as Murmur3.

Next, we characterize its regret in the setting discussed in \Cref{sec:allocators}, when keys are allocated in a tight loop but do not fit on a single page.
The keys are pointers to objects distributed uniformly between multiple pages and within pages form subsets of values of arithmetic sequences of the same increment.

\begin{restatable}{proposition}{propexpectedcostofpointermix}[Expected Cost of Pointer-Mix]
Let $k_{1:n}$ be integer keys, and $\mathcal{P}=\{k_i \gg \mathrm{PB}\ \colon i \in [1,n]\}$ the set of pages (the high bits of keys).
Let the keys be distributed over the pages uniformly, $n=|\mathcal{P}|u$, where $u$ is the number of keys on the same page ($u=|\{i\colon k_i \gg \mathrm{PB} = p\}|$ for all pages $p \in \mathcal{P}$).
We assume that all $u$ keys on the same page form random subsets of arithmetic progressions with page specific offsets but the same increments.
Then, the expected cost of the Pointer-Mix hash function is
\begin{align*}
1+\frac{n - u \min\pigl(1,\frac{2^{\mathrm{PB} - s}}{m}\pigr)}{2m}.
\end{align*}
\end{restatable}
\noindent See \Cref{sec:proofs} for the proof.
This cost is upper-bounded by that of the uniform hash (\Cref{prop:expected-cost-of-uh}), $1+\frac{n-1}{2m}$, and we can see that with a few densely packed pages, we can get reasonable improvements, which diminish quickly with more pages and sparsity.
Meanwhile, at the single-page extreme ($u=n\leqslant m$), Pointer-Mix is a perfect hash.

\subsection{The Pointer-Shift Hash}
\label{sec:the-pointer-shift-hash}

We have seen that as the number of pages grows, Pointer-Mix quickly falls back to the performance level of the uniform hash.
It is also slow: it includes a general purpose hash.

In practice, we found that the Pointer-Shift hash $k \rightarrow k \gg s' + k \gg \mathrm{PB}$ often outperforms Pointer-Mix.
Furthermore, it also behaves rather similarly to the Arithmetic hash if $s'$ and $PB$ are not close in value.
To avoid degenerating to $k \rightarrow 2 k \gg \mathrm{PB}$ at $s=\mathrm{PB}$, $s'$ is set to a large value in this case to zero out the first term without introducing a slow conditional.

\subsection{The Constant Hash}

The hash functions discussed up to now are adaptive, as they all involve the shift $s$ detected from the keys.
A different kind of adaptation, based on the number of keys but ignoring their values is also possible.
We mirror the common practice of starting with an array plus linear search and switching to a hash table above a predefined, small number of keys but hide it behind the hash table API and utilize it when the comparison function is extremely light as is the case with integer / pointer hashing.
This may also be viewed as a \emph{Constant} hash with a specialized implementation or a hash table with only one bucket.

\subsection{Other Hashes}

Stock SBCL comes with the \emph{Prefuzz} hash function, which was designed by hand and performs well empirically in many situations.
Naturally, Prefuzz takes advantage of the memory allocation patterns to make common use-cases fast (e.g. \lisp{symbol} keys, frequently used in the compiler), but it does so at the expense of extreme penalties to others.

The \emph{Murmur3} mixer function is a fast, widely used, general-purpose, non-cryptographic hash function with strong mixing properties.
Its bucket distribution is very close to that of the uniform hash.

\begin{algorithm}[t]
  \caption{\small Adapting the hash function in \lisp{eq} (i.e. object identity based) hash tables at rehash.
Note that we count collisions with Prefuzz only at larger sizes; otherwise it adapts only through the max-chain-length mechanism (see \Cref{alg:puthash}) to reduce the overhead.
We refer to this algorithm as Co+PS>Pr>Mu when comparing minor variations.}
  \label{alg:rehash-eq}
  \hrule
  \begin{algorithmic}[1]
    \Require{Integer/pointer keys $k_i (i \in [1, \dots, n])$, doubled number of buckets $m$, current hash function $h$.}
    \Procedure{$\operatorname{rehash\_and\_maybe\_adapt\_eq}$}{}
      \If{h = constant\_hash}
        \If{m = 64}
          \Let{$s$}{$\operatorname{count\_common\_prefix\_bits}(k_{1:16}$)}
          \Let{$h$}{$\operatorname{pointer\_shift}$}
        \EndIf
      \EndIf
      \If{$h = \operatorname{pointer\_shift}$}
        \Let{$\mathit{n\_collisions}$}{$\operatorname{rehash}(m, h, \mathit{count=True})$}
        \If{$\mathit{n\_collisions}$ is too many}
          \Let{$h$}{$\operatorname{prefuzz}$}
        \EndIf
      \EndIf
      \If{$h = \operatorname{prefuzz}$}
        \If{$m < 2048$}
          \State{$\operatorname{rehash}(m, h)$}
        \Else
          \Let{$\mathit{n\_collisions}$}{$\operatorname{rehash}(m, h, \mathit{count=True})$}
          \If{$\mathit{n\_collisions}$ is too many}
            \Let{$h$}{$\operatorname{murmur3}$}
          \EndIf
        \EndIf
      \EndIf
      \If{$h = \operatorname{murmur3}$}
        \State $\operatorname{rehash}(m, h)$
      \EndIf
    \EndProcedure
  \end{algorithmic}
  \hrule
\end{algorithm}

\subsection{Adapting the Hash Function}
\label{sec:choosing-the-hash-function}

We implemented the above hash functions in SBCL and modified its hash table implementation to perform adaptation with pointers and integers.
So, we use Common Lisp's \lisp{eq} hash tables, whose comparison function is based on object identity (i.e. addresses of non-immediate objects or the values of immediate objects such as integers that fit into a machine word) similarly to the \texttt{==} operator in Java.

The adaptation mechanism for \lisp{eq} hash tables (\Cref{alg:rehash-eq}) fleshes out \Cref{alg:puthash} and works as follows.
Hash tables are initialized with the Constant hash.
When the number of keys exceeds 32, $\mathit{count\_common\_prefix\_bits}$ determines the shift $s$ to use, and we switch to the Pointer-Shift hash function.
Whenever rehash finds that there are significantly more collisions than would be expected with a uniform hash (see \Cref{sec:number-of-collisions}), the hash table switches to Prefuzz.

\newcommand\figfixnumprogoneregret{
\tikzpicturedependsonfile{data/fixnum-prog-1-murmur}
\tikzpicturedependsonfile{data/fixnum-prog-1-prefuzz}
\tikzpicturedependsonfile{data/fixnum-prog-1-flat}
\tikzpicturedependsonfile{data/fixnum-prog-1-flat-safe-small}
\begin{tikzpicture}
  \begin{axis}[xlabel=number of keys, ylabel=regret,
      xmin=1, xmode=log, log basis x={2},
      ymin=-0.01, ymax=0.79,
      legend pos=north west,
      legend columns=2,
      legend style={nodes={scale=0.7, transform shape}},
      legend cell align={left},
      height=0.6*\linewidth,
      width=0.98\linewidth,
    ]
    \pgfplotstableread{data/fixnum-prog-1-murmur.tbl}{\sorted}
    \addplot[uniformhash] table [x=nkeys, y=rndregret] {\sorted};
    \addlegendentry{Uniform};

    \pgfplotstableread{data/fixnum-prog-1-murmur.tbl}{\sorted}
    \addplot[murmurhash] table [x=nkeys, y=regret] {\sorted};
    \addlegendentry{Murmur};

    \pgfplotstableread{data/fixnum-prog-1-prefuzz.tbl}{\sorted}
    \addplot[prefuzzhash] table [x=nkeys, y=regret] {\sorted};
    \addlegendentry{Prefuzz};

    \pgfplotstableread{data/fixnum-prog-1-flat.tbl}{\sorted}
    \addplot[constanthash] table [x=nkeys, y=regret] {\sorted};
    \addlegendentry{Co+Pr};

    \pgfplotstableread{data/fixnum-prog-1-flat-safe-small.tbl}{\sorted}
    \addplot[adaptivehash] table [x=nkeys, y=regret] {\sorted};
    \addlegendentry{\textbf{Adaptive}};
    \addplot[mark=*,mark options={scale=1.0}] coordinates {(33,0)};
  \end{axis}
\end{tikzpicture}
\caption{Regret with \texttt{FIXNUM} \texttt{:PROG 1}.
Murmur closely tracks Uniform.
Prefuzz is aggressively optimized for small sizes.
Adaptive (\Cref{alg:rehash-eq}) is a perfect hash here.
Both Co+Pr (Constant followed by Prefuzz) and Adaptive use the Constant hash until the fixed switch point at 32 keys (black dot).}
}

\newcommand\figfixnumprogoneput{
\tikzpicturedependsonfile{data/fixnum-prog-1-murmur}
\tikzpicturedependsonfile{data/fixnum-prog-1-prefuzz}
\tikzpicturedependsonfile{data/fixnum-prog-1-flat}
\tikzpicturedependsonfile{data/fixnum-prog-1-flat-safe-small}
\begin{tikzpicture}
  \begin{axis}[ylabel=ns / put,
      xmin=1, xmode=log, log basis x={2},
      ymin=20, ymax=180, ymode=log, log basis y={2},
      legend pos=north west,
      legend style={nodes={scale=0.7, transform shape}},
      legend cell align={left},
      height=0.6*\linewidth,
      width=0.98\linewidth,
    ]
    \pgfplotstableread{data/fixnum-prog-1-murmur.tbl}{\sorted}
    \addplot[murmurhash] table [x=nkeys, y=putns] {\sorted};
    \addlegendentry{Murmur};

    \pgfplotstableread{data/fixnum-prog-1-prefuzz.tbl}{\sorted}
    \addplot[prefuzzhash] table [x=nkeys, y=putns] {\sorted};
    \addlegendentry{Prefuzz};

    \pgfplotstableread{data/fixnum-prog-1-flat.tbl}{\sorted}
    \addplot[constanthash] table [x=nkeys, y=putns] {\sorted};
    \addlegendentry{Co+Pr};

    \pgfplotstableread{data/fixnum-prog-1-flat-safe-small.tbl}{\sorted}
    \addplot[adaptivehash] table [x=nkeys, y=putns] {\sorted};
    \addlegendentry{\textbf{Adaptive}};
  \end{axis}
\end{tikzpicture}
\caption{PUT timings with \texttt{FIXNUM} \texttt{:PROG 1}.
Prefuzz outperforms Murmur even at large sizes despite higher regret because it's friendlier to the cache (its collisions are between subsequent elements of the progression), and its combination with Constant is even faster.
Thus, despite being a perfect hash, Adaptive can improve on them only marginally.}
}

\newcommand\figfixnumprogoneget{
\begin{tikzpicture}
  \begin{axis}[ylabel=ns / get,
      xmin=1, xmode=log, log basis x={2},
      ymin=6, ymax=180, ymode=log, log basis y={2},
      legend pos=north west,
      legend style={nodes={scale=0.7, transform shape}},
      legend cell align={left},
      height=0.6*\linewidth,
      width=0.98\linewidth,
    ]
    \tikzpicturedependsonfile{data/fixnum-prog-1-murmur}
    \tikzpicturedependsonfile{data/fixnum-prog-1-prefuzz}
    \tikzpicturedependsonfile{data/fixnum-prog-1-flat}
    \tikzpicturedependsonfile{data/fixnum-prog-1-flat-safe-small}

    \pgfplotstableread{data/fixnum-prog-1-murmur.tbl}{\sorted}
    \addplot[murmurhash] table [x=nkeys, y=getns] {\sorted};
    \addlegendentry{Murmur};

    \pgfplotstableread{data/fixnum-prog-1-prefuzz.tbl}{\sorted}
    \addplot[prefuzzhash] table [x=nkeys, y=getns] {\sorted};
    \addlegendentry{Prefuzz};

    \pgfplotstableread{data/fixnum-prog-1-flat.tbl}{\sorted}
    \addplot[constanthash] table [x=nkeys, y=getns] {\sorted};
    \addlegendentry{Co+Pr};

    \pgfplotstableread{data/fixnum-prog-1-flat-safe-small.tbl}{\sorted}
    \addplot[adaptivehash] table [x=nkeys, y=getns] {\sorted};
    \addlegendentry{\textbf{Adaptive}};
  \end{axis}
\end{tikzpicture}
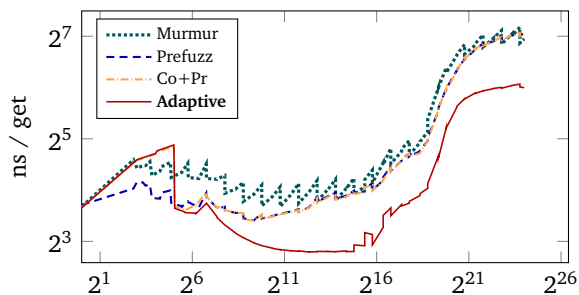
\caption{GET timings with \texttt{FIXNUM} \texttt{:PROG 1}.
Keys are queried in random order so regret matters more here than with PUT, but the cache-friendliness of Prefuzz still keeps it ahead of Murmur.
As expected, Adaptive can finally benefit from its zero regret after the Constant hash phase.}
}

\begin{figure}[t]
\figfixnumprogoneregret
\label{fig:fixnum-prog-1-regret}
\end{figure}

\begin{figure}[t]
\figfixnumprogoneput
\label{fig:fixnum-prog-1-put}
\end{figure}

We also need to fall back on Murmur3 in case Prefuzz produces too many collisions.
However, because Prefuzz is somewhat robust and considerably faster than Murmur3, we count collisions only at larger sizes to reduce the adaptation overhead\footnote{The overhead is that of incrementing a single counter if the bucket in \lisp{index-vector} is not zero, which is a hard to predict branch for the CPU.} and otherwise rely on the max-chain-length mechanism from \Cref{alg:puthash} to fall back on the next safer hash function (in the order they appear in $rehash\_and\_maybe\_adapt\_eq$ in \Cref{alg:rehash-eq}) if the key being inserted falls into a bucket with at least 14 other keys\footnote{With a uniform hash, the probability of the maximum chain length being at most 14 is higher than 99\% for all possible hash table sizes.}.
Using the tighter limit from \Cref{sec:half-an-example} would be too costly for \lisp{eq} hashing.
See \Cref{sec:implementation-details} for the more details.

\begin{figure}[t]
\figfixnumprogoneget
\label{fig:fixnum-prog-1-get}
\end{figure}

In summary, we have two ways of detecting when the current hash function is likely suboptimal: tracking collisions at rehash and chain length at insertion.
We use collision tracking to catch gradual degradations of performance, and max-chain-length to catch catastrophic failures in a single bucket.
One can construct lower and upper bounds on the average cost of lookup based on the collision count and max-chain-length.
Neither of these mechanisms are perfect, but they work quite well in tandem to inform the adaptation logic about the cost of lookup.

Both the collision count and max chain length are proxies for the cost of lookup as defined in \Cref{def:cost-of-hashes}, which in turn is a proxy for performance that ignores important factors such as cache effects.
Thus, even if adaptation can be driven by proxy statistics, there is no way around actually measuring performance directly.

\subsection{Microbenchmarks}
\label{sec:microbenchmarks}

We conducted experiments on SBCL's \lisp{eq} hash tables with various hash functions, hash table sizes, integer and pointer keys.
Our experimental methodology is the same as in \Cref{sec:experiments-with-string-keys} except for how keys are generated, which we briefly describe below (see \Cref{sec:microbenchmarking-methodology} for the details).
In the experiments, Co+Pr starts with the Constant hash and switches to Prefuzz above 32 keys.
Adaptive behaves as described in \Cref{alg:rehash-eq}.

\begin{itemize}
\item{(\texttt{FIXNUM} \texttt{:PROG 1})}
The first experiment is with \lisp{fixnum}s\footnote{A fixnum is a signed integer in Lisp that fits into a machine word. It's similar to an \texttt{int64} on x86-64.} following an arithmetic progression with increment 1 (denoted by \texttt{:PROG 1}).
As \Cref{fig:fixnum-prog-1-regret} shows, the regrets of Murmur and Uniform are very close, as expected.
Prefuzz seems to be aggressively optimized for small hash table sizes, and Adaptive is a perfect hash in this simple scenario.
However, differences in regret do not predict actual performance well, which is most visible in \Cref{fig:fixnum-prog-1-put}, where insertion is much faster with Prefuzz than with Murmur even where the latter has much smaller regret.
This is because the collisions with Prefuzz are between keys close in insertion order, which benefits the CPU's cache.
The same effect is present in lookups (\Cref{fig:fixnum-prog-1-get}), although to a lesser degree because the benchmark looks up keys in random order, so some random access to memory is inevitable.
See \Cref{sec:results-for-fixnum-prog-1} for the full set of results.

\item{(\texttt{FLOAT} \texttt{:PROG 1})}
Similar to the previous \lisp{fixnum} case, we also tested \lisp{single-float} keys.
As the regret curves in \Cref{fig:float-prog-1-regret} show, Prefuzz suffers a catastrophic failure, placing most keys in the same bucket, but Adaptive takes advantage of the many constant low bits in the keys, eventually falling back to Prefuzz and then to Murmur.
See \Cref{sec:results-for-float-prog-1} for the full set of results.

\item{(\texttt{FIXNUM} \texttt{:PROG 12})}
Next, we tested arithmetic progressions with increment 12.
This is intended to test whether the shift detection in \Cref{alg:detect-shift} works.
The regret curves in \Cref{fig:fixnum-prog-12-regret} and the operation timings tell a similar story as with \texttt{:PROG 1} except that Adaptive is no longer a perfect hash due to Pointer-Shift's $k \gg \mathrm{PB}$ term.
See \Cref{sec:results-for-fixnum-prog-12} for more.

\item{(\texttt{FIXNUM} \texttt{:RND 6})}
Like \texttt{:PROG 6}, but keys follow a random progression: 0--5 values are skipped randomly between subsequent keys.
This is intended to approximate populating a hash table with non-uniformly sized keys or values being allocated in a tight loop.
Results in \Cref{sec:results-for-fixnum-rnd-6} shows that Prefuzz is better than Murmur, and the large gains made by the Constant hash persist, but with less structure in the key sets, Prefuzz is harder to beat.

\item{(\texttt{CONS} \texttt{:RND 6})}
Similar to \texttt{FIXNUM} \texttt{:RND 6}, but we allocate real \lisp{cons} objects.
See \Cref{sec:results-for-cons-rnd-6} for the results.

\item{(\texttt{SYMBOL} \texttt{:EXISTING})}
Finally, we explore the case with the set of existing \lisp{symbol}s from Lisp as keys.
\Cref{sec:results-for-symbol-existing} shows that despite the scarcity of structure in the key distribution, Prefuzz maintains a small advantage over Murmur.
As expected, Co+Pr and Adaptive follow suit, most of their advantage being in the Constant hash phase.
\end{itemize}

\newcommand\figfloatprogoneregret{
\begin{tikzpicture}
  \begin{axis}[xlabel=number of keys, ylabel=regret,
      xmin=1, xmode=log, log basis x={2},
      ymin=-0.01, ymax=2000, ymode=log, log basis y={10},
      legend pos=south east,
      legend columns=2,
      legend style={nodes={scale=0.7, transform shape}},
      legend cell align={left},
      height=0.6*\linewidth,
      width=0.98\linewidth,
    ]
    \pgfplotstableread{data/float-prog-1-murmur.tbl}{\sorted}
    \addplot[uniformhash] table [x=nkeys, y=rndregret] {\sorted};
    \addlegendentry{Uniform};

    \pgfplotstableread{data/float-prog-1-murmur.tbl}{\sorted}
    \addplot[murmurhash] table [x=nkeys, y=regret] {\sorted};
    \addlegendentry{Murmur};

    \pgfplotstableread{data/float-prog-1-prefuzz.tbl}{\sorted}
    \addplot[prefuzzhash] table [x=nkeys, y expr={min(\thisrow{regret}, 2000)}] {\sorted};
    \addlegendentry{Prefuzz};

    \pgfplotstableread{data/float-prog-1-flat.tbl}{\sorted}
    \addplot[constanthash] table [x=nkeys, y expr={min(\thisrow{regret}, 2000)}] {\sorted};
    \addlegendentry{Co+Pr};

    \pgfplotstableread{data/float-prog-1-flat-safe-small.tbl}{\sorted}
    \addplot[adaptivehash] table [x=nkeys, y expr={min(\thisrow{regret}, 2000)}] {\sorted};
    \addlegendentry{\textbf{Adaptive}};
    \addplot[mark=*,mark options={scale=1.0}] coordinates {(33,0)};
  \end{axis}
\end{tikzpicture}
\caption{Regret with \texttt{FLOAT} \texttt{:PROG 1}.
To be able to plot the catastrophic failure of Prefuzz (and of Co+Pr, consequently), we use log scale for regret on this graph.
Single floats are especially problematic for Prefuzz because they can have many constant low bits.
Adaptive detects these constant low bits and does better than Uniform until variation in the floating point exponents makes its original estimate of the number of constant low bits invalid, and the resulting gradual increase in collisions makes it switch to Prefuzz at rehash time.
This is a spectacularly bad idea in this scenario, and the high number of collisions causes an immediate switch to Murmur.
The switch times vary by hash table because the key sets are generated starting from random offsets.}
}

\newcommand\figfixnumprogtwelveregret{
\begin{tikzpicture}
  \begin{axis}[xlabel=number of keys, ylabel=regret,
      xmin=1, xmode=log, log basis x={2},
      ymin=-0.01, ymax=0.79,
      legend pos=north west,
      legend columns=2,
      legend style={nodes={scale=0.7, transform shape}},
      legend cell align={left},
      height=0.6*\linewidth,
      width=0.98\linewidth,
    ]
    \pgfplotstableread{data/fixnum-prog-12-murmur.tbl}{\sorted}
    \addplot[uniformhash] table [x=nkeys, y=rndregret] {\sorted};
    \addlegendentry{Uniform};

    \pgfplotstableread{data/fixnum-prog-12-murmur.tbl}{\sorted}
    \addplot[murmurhash] table [x=nkeys, y=regret] {\sorted};
    \addlegendentry{Murmur};

    \pgfplotstableread{data/fixnum-prog-12-prefuzz.tbl}{\sorted}
    \addplot[prefuzzhash] table [x=nkeys, y=regret] {\sorted};
    \addlegendentry{Prefuzz};

    \pgfplotstableread{data/fixnum-prog-12-flat.tbl}{\sorted}
    \addplot[constanthash] table [x=nkeys, y=regret] {\sorted};
    \addlegendentry{Co+Pr};

    \pgfplotstableread{data/fixnum-prog-12-flat-safe-small.tbl}{\sorted}
    \addplot[adaptivehash] table [x=nkeys, y=regret] {\sorted};
    \addlegendentry{\textbf{Adaptive}};
    \addplot[mark=*,mark options={scale=1.0}] coordinates {(33,0)};
  \end{axis}
\end{tikzpicture}
\caption{Regret with \texttt{FIXNUM} \texttt{:PROG 12}.
Murmur closely tracks Uniform, but Prefuzz is better across almost the whole range.
Arithmetic (\Cref{sec:the-arithmetic-hash}) would be a perfect hash here, but Adaptive, which uses Pointer-Shift (\Cref{sec:the-pointer-shift-hash}), is not quite perfect due to the interference of its $k \gg \mathrm{PB}$ term.}
}

\begin{figure}[t]
\figfloatprogoneregret
\label{fig:float-prog-1-regret}
\end{figure}

\begin{figure}[t]
\figfixnumprogtwelveregret
\label{fig:fixnum-prog-12-regret}
\end{figure}

In all results presented, the Co+Pr and Adaptive hash functions, which both start out with the Constant hash, gain a lot of performance on insertion but lose on lookups.
This is still an overall win based just on the numbers presented here except in very lookup-heavy workloads.
However, an even larger unquantified benefit is in the reduced memory usage and garbage collection times due to the Constant hash having a specialized single vector implementation.

In summary, a general-purpose hash such as Murmur is a safe but suboptimal choice for many common situations in \lisp{eq} hash tables.
SBCL's own Prefuzz hash is hand-crafted for these common situations, where it is difficult to beat.
Still, because it is non-adaptive, it sacrifices performance in other cases and has terrible worst-case behaviour.
The adaptive approach combines the worst-case safety of Murmur with the common-case performance of Prefuzz, and it even manages to slightly outperform Prefuzz because having the reliable safety net of the fallback mechanism allows it to be more aggressive in catering to the common case.

\subsection{Macrobenchmarks}
\label{sec:macrobenchmarks}

Microbenchmarks are useful indicators of the highest achievable throughput, but their results do not necessarily carry over to more complex workloads, where factors such as code size and complexity gain importance.
To investigate this issue, we conducted experiments where hash table operations constitute only a small fraction of the workload.
In particular, we measured the times to 1.~compile and load a set of libraries; 2.~run the tests of the same set of libraries; 3.~run each test file in SBCL's \texttt{tests/} directory.
The detailed results are in \Cref{sec:macrobenchmark-results}; here, we only provide a summary.
\begin{enumerate}
\item The first suite is the heaviest on \lisp{eq} and \lisp{equal} hash table operations.
With SBCL's statistical profiler, we estimated that about 1.7\% of the total runtime was spent in small \lisp{eq} hash tables (that is, within the Constant hash's range of 0--32 keys) and 1.3\% in larger ones, while operations on \lisp{equal} hash tables took 1\%.
The relative speedup was 8\% in the large \lisp{eq} case and 50\% for \lisp{equal}.
The gain in the small \lisp{eq} case is harder to pin down because the Constant hash's significantly reduced garbage collection cost; we estimate it to lie in the 8\%--14\% range.

\item In the second suite, small/large \lisp{eq} and \lisp{equal} hash table operations constituted 0.55\%, 0.25\% and 0.3\% of the baseline result.
Relative gains were as previously except for the large \lisp{eq} case, possibly because our benchmarking methodology is not able to detect differences so small, or because the increased code size is not worth it in code paths so cold.

\item Finally, timing SBCL tests on a per-file basis revealed no major performance regressions.
Overall, we observed a 0.7\% gain due to adaptive \lisp{eq} hashing, with adaptive \lisp{equal} hashing contributing another 1.5\%, which came almost exclusively from the two tests with longer list keys (see \Cref{sec:experiments-with-list-keys}).
\end{enumerate}

In summary, by surviving the difficult transition from hot-path microbenchmarking to the cooler workloads reported in this section, adaptive hashing emerges as a method of practical relevance.


\section{Related Works}

Our method takes inspiration from Perfect Hashing, which selects a hash function for a given set of keys (known as static hashing).
Dynamic Perfect Hashing \citep{dietzfelbinger1994dynamic,belazzougui2009hash} allows the set of keys to change but still guarantees worst-case constant lookup time.
However, it requires more memory than plain hash tables, so it is not a drop-in replacement for them.
Cuckoo hashing also has worst-case constant lookup time but with a lower memory footprint.
In a sequential implementation, it requires 1.5 hash function evaluations and memory accesses on average per lookup, which is about what the uniform hash has at load factor 1 (see \Cref{prop:expected-cost-of-uh}).

Although Rabbit Hashing \citep{rabbithashing} and the hash table by \citet{iwrotehtefastesthashtable} utilize maximum probe length to determine when to resize or reseed, unlike our adaptive method, they stop short of modifying the functional form of the hash, nor do they tailor it to fit a specific set of keys.

VIP hashing \citep{kakaraparthy2022vip} moves frequently accessed keys earlier in the collision chains to reduce the average lookup cost, adapting the storage layout to the access pattern.
As with our approach, this adaptation is performed online and requires careful management to minimize runtime overhead.
However, in contrast to our work, VIP hashing leaves the hash function itself constant.

\Citet{hentschel2022entropy} propose to learn the hash function offline from samples of the key distribution by using the most informative parts of compound keys.
Our case study on string keys in \Cref{sec:half-an-example} can be seen as an online version of their method, with the expensive learning phase removed.

In the taxonomy of \citet{chi2017hashing}, adaptive hashing falls under data-dependent hashing although they assume that training is performed offline.
The adaptive hashing method can also be viewed as a more robust, key-aware version of user-defined hash functions, which are also adapted offline to a particular key distribution.

There is a history of handcrafted hash functions that perform well in common cases, but exhibit spectacular failures in others.
As we have seen, Prefuzz in SBCL is one such example.
Java used to hash only about 1/8 of the characters in long strings \citep{jdkstringhashsubsampling}.
This hardcoded limit made hashing faster, but as it could lead to lots of collisions with no fallback mechanism to save it, from JDK 1.2 on, all characters are hashed.


\section{Conclusion}

We have laid out the case for adaptive hash functions, which reside between key-agnostic and perfect hashes.
Our primary contribution lies in reconceptualizing hash tables as inherently adaptive data structures, which can marry the theoretical guarantees of universal hashing with the common-case performance of weak hash functions.
To validate this approach, we implemented the adaptive framework and demonstrated improved performance as well as robustness on string and integer/pointer hashing by capturing real-life key patterns and providing efficient search algorithms.
The design space afforded by the adaptive hashing framework is large, and the adaptation mechanisms investigated in this work hardly cover a substantial or particularly imaginative part of it, leaving ample room for further developments.
In particular, the max-chain-length mechanism can form the basis of a defense against denial-of-service collision attacks without constraining the choice of hash function \citep{aumasson2012siphash}.

Finally, the source code of the SBCL changes and the benchmarking code to reproduce the experimental results are open-sourced and available at \url{https://github.com/melisgl/sbcl/tree/adaptive-hash}.

\section*{Acknowledgements}
We thank Christophe Rhodes, Miloš Stanojević, Andrew Senior, Paul-Virak Khuong, and the reviewers for their valuable comments.

\bibliographystyle{plainnat}
\bibliography{paper}

\appendix


\newpage
\section{Proofs}
\label{sec:proofs}

We restate propositions and provide proofs.

\propexpectedcostofuh*
\begin{proof}
Because we sample a hash function independently for each key, $\pi_i(k_i)$ are independent and distributed uniformly over the key space.
Writing the expected number of comparisons for a lookup as a sum over the $i$th key added to the hash table, we get
\begin{equation*}
\frac{\sum_{i=0}^{n-1} \big(1 + \frac{i}{m}\big)}{n} = \frac{n + \frac{n(n-1)}{2m}}{n} = 1+\frac{n-1}{2m}.
\qedhere
\end{equation*}
\end{proof}

\propperfecthasheshaveminimalcost*
\begin{proof}
A set of hash values is either perfect or non-perfect.
Since all perfect hashes have the same cost $C(U(n,m))$, if we construct a perfect hash with lower cost from any non-perfect hash, it follows that $C(U(n,m))$ is minimal.

Next, we show one such construction.
For any non-perfect set of hash values with bucket counts $c$, there are always two buckets $i$ and $j$ such that $c_i > c_j + 1$ (else it would be a perfect hash due to bucket counts having to sum to $n$).
By moving one hash from bucket $i$ to $j$, we get a new set of hash values with bucket counts $c'$ whose cost is lower because $2n(C(c)-C(c'))=c_i(c_i+1) + c_j(c_j+1) - (c_i-1)c_i - (c_j+1)(c_j+2) = c_i - c_j - 1 > 0$.
\end{proof}

\propexpectedregretofuh*
\begin{proof}
Let $\pi_{1:qm}(k_{1:qm}) = [\pi_1(k_1), \dots, \pi_{qm}(k_{qm})]$.
Then,
\begin{align*}
\E_{\pi_{1:qm}} &R(\pi_{1:qm}(k_{1:qm}), m)\\
&= \E_{\pi_{1:qm}}\bigl[C\bigl(c(\pi_{1:qm}(k_{1:qm}),m)\bigr) - C\bigl(U(n,m)\bigr)\bigr]\\
&= 1+\frac{qm-1}{2m} - \frac{q+1}{2}\\
&= 0.5 + \frac{1}{m},
\end{align*}
which concludes the proof.
\end{proof}

\propexpectedcostofpointermix*
\begin{proof}
First, we look at the case where there can be no collisions between keys on the same page.
This is true if the keys form an arithmetic progression and are not just random subsets.
Due to the subset assumption, it is also true if the hash table is large enough to hold a page worth of keys (shifted by $s$): $log_2(m) \geqslant \mathrm{PB} - s$.

The cost decomposes as the sum of the number of hashes in the same bucket as keys are added one by one.
Thus, when there are $(p-1)$ previous pages' worth of keys already in the hash table, all $u$ keys on the next page will contribute the same amount $1+(p-1)\frac{u}{m}$ to the cost because there are no collisions between them.
\begin{align*}
C &= n^{-1}\sum_{p=1}^{\abs{\mathcal{P}}} u\Pigl(1+(p-1)\frac{u}{m}\Pigr)
= 1+\frac{n-u}{2m}.
\end{align*}

Second, if keys (shifted by $s$) on the same page may collide randomly modulo $m$, then we can expect to get only $um^{-1}2^{\mathrm{PB} - s}$ guaranteed no colliding keys.
Updating our formula, we get that
\begin{align*}
C &= 1+\frac{n - u \min\pigl(1,\frac{2^{\mathrm{PB} - s}}{m}\pigr)}{2m}.
\qedhere
\end{align*}
\end{proof}


\section{The Expected Number of Collisions}
\label{sec:number-of-collisions}

Here, we derive computationally cheap upper bounds on the expected number of collisions with the uniform hash for testing whether the observed number of collisions in \Cref{alg:puthash} is too many.

Given $n$ keys, $m$ buckets, and hash values $h_{1:n}$, let $c$ and $u$ be the number of collisions and unused (empty) buckets, respectively.
The number of used buckets $m-u$ is equal to the number of non-colliding keys $n-c$, so $c=u+n-m$ and it suffices to bound $u$ to bound $c$.
Appealing to the birthday problem \citep{siegrist1197random}, the expected number of unused buckets is
\begin{align*}
\smash[t]{\E u = m \biggl(1-\frac{1}{m}\biggr)^n}.
\end{align*}
Holding the load factor $f=n/m$ constant, the proportion of unused buckets is monotonically increasing in $m$ and
\begin{align*}
\lim_{m \to \infty} \E \frac{u}{m} = \lim_{m \to \infty} \E \piggl(\biggl(1-\frac{1}{m}\biggr)^m\piggr)^{n/m} = \exp\biggl(-\frac{n}{m}\biggr) = e^{\neg f},
\end{align*}
using the product limit formula of the exponential function.
As the first three plots in \Cref{fig:collision-bounds} show, this asymptotic formula is a very tight upper bound even for small $m$.

\begin{figure}[t]
\begin{tikzpicture}
  \begin{axis}[xmin=0, xmax=1, ymax=1, xlabel=load factor $f$,
      ylabel=proportion of empty buckets,
      height=0.75*\linewidth,
      width=0.98\linewidth,
      legend columns=2,
      legend pos=south west]
    \addplot[mglred, thick, domain=0:1] {0.875^(8*x)};
    \addlegendentry{$m=8$};
    \addplot[orange, dashed, thick, domain=0:1] {0.9375^(16*x)};
    \addlegendentry{$m=16$};
    \addplot[mglgreen, densely dashdotted, thick, domain=0:1] {e^(-x)};
    \addlegendentry{$e^{\neg f}$};
    \addplot[mglgreen!50!cyan, loosely dashdotted, thick, domain=0:1] {e^(-x)/0.9};
    \addlegendentry{$e^{\neg f} / 0.9$};
    \addplot[cyan, densely dotted, thick, domain=0:1] {1-9/16*x};
    \addlegendentry{$1-\nicefrac{9}{16}f$};
      legend style={font=\small}]
    \addplot[mglblue, dotted, thick, domain=0:1] {1-0.5*x};
    \addlegendentry{$1-f/2$};
  \end{axis}
\end{tikzpicture}
\caption{Expected proportion of empty buckets with the uniform hash given a load factor (always in $[0,1]$) with the two smallest possible hash table sizes ($m=8$ and $m=16$), the tight upper bound $e^{\neg f}$.
The $e^{\neg f} / 0.9$ adds some slack, while $1-\nicefrac{9}{16}f$ and $1-f/2$ provide increasingly looser but cheaper to compute upper bounds for small hash tables.}
\label{fig:collision-bounds}
\end{figure}

However, $e^{\neg f}$ involves a division and an evaluation of the exponential function.
This overhead is measurable at small hash table sizes.
It is also at small sizes that the hash table is more likely to fit in the cache; thus, the cost of collisions is also smaller, and we may allow more slack.
In practice, we test for too many collisions with the following conditions:
\begin{itemize}
\item $1-f/2 < \frac{u}{m}$ for $m < 1024$;
\item $1-\nicefrac{9}{16}f < \frac{u}{m}$ for $1024 \leq m < 4096$;
\item $e^{\neg f} / 0.9 < \frac{u}{m}$ for $4096 \leq m$.
\end{itemize}
Crucially, $1 - f/2 < \frac{u}{m}$ can be simplified as
\begin{align*}
1 - f/2 &< \frac{u}{m} \\
m - n/2 &< u \\
m - n/2 &< c + m - n \\
n/2 &< c,
\end{align*}
which can be implemented with a single arithmetic shift and an integer comparison.
Similarly, $1-\nicefrac{9}{16}f < \frac{u}{m}$ simplifies to $\nicefrac{7}{16} n < c$, which is equivalent to $(n \gg 1) - (n \gg 4) < c$, involving only two arithmetic shifts.


\section{SBCL Hash Tables}
\label{sec:sbcl-hash-tables}

SBCL hash tables are technically separate-chaining, meaning that there are explicit chains of keys which fall into the same bucket.
Ironically for a Lisp, these chains are not lists: for performance reasons, they are represented by two arrays of indices, called the \lisp{index-vector} and the \lisp{next-vector}, which are only resized when the hash table grows.
The main pieces fit together as follows:
\begin{itemize}
\item The vector \lisp{pairs} holds alternating keys and values in a stable order, which is important for iteration (e.g. \lisp{maphash}).
The first key is at index 2.
\item \lisp{index-vector} is a power-of-2-sized array of indices, that maps a bucket to the index of the first key--value pair in \lisp{pairs} or zero if it's empty.
\item \lisp{next-vector} maps the index of a pair to the index of the next pair in the collision chain or zero at chain end.
It also chains empty slots in \lisp{pairs} together.
\item For all but the lightest hash functions (standard \lisp{eq} and \text{eql} tables), the hash values of all keys in the hash table are cached in \lisp{hash-vector}.
At lookup, the cached hash is compared to the hash of the key being looked up, and if they are different, then we know without invoking the potentially expensive comparison function that they cannot match (\Cref{alg:puthash}).
\end{itemize}
An important operation is \emph{rehashing}.
When the number of buckets increases, keys are reassigned (``rehashed'') to buckets based on their hash values, which are taken from \lisp{hash-vector} if there is one.
Rehashing iterates over \lisp{pairs}, rewriting \lisp{index-vector} and \lisp{next-vector}.

Each standard hash table type (\lisp{eq}, \lisp{eql}, \lisp{equal}, \lisp{equalp}) have separate accessors (GET, PUT, DEL), which are invoked through an indirect call, and the two lighter ones have the hash function and comparison function inlined.
There is no SIMD, and SBCL does not devirtualize calls.


\section{Implementation Details}
\label{sec:implementation-details}

\lisp{Eq} hash tables are initialized with the Constant hash and a \lisp{pairs} vector is allocated for up to 8 key--value pairs.
At this time, there is no \lisp{index-vector} and \lisp{next-vector}.
The \lisp{pairs} vector is doubled in size as more keys are added, but no rehashing is necessary as there are no chains yet.
Once the number of keys exceeds 32, $\mathit{count\_common\_prefix\_bits}$ (see \Cref{alg:detect-shift}) is invoked with 16 keys to guess the shift $s$, we switch to the normal SBCL hash table implementation, set the hash function to $\mathit{pointer\_shift}$ and rehash (\Cref{alg:rehash-eq}).
Switching to the new hash function is implemented as replacing the GET, PUT, DEL implementations via their accessors.

We added new hash table accessors for Murmur and also a new slot for the detected shift $s$ to the hash table structure.
Instead of adding separate accessors for Pointer-Shift, which actually had the best performance in microbenchmarks, we made Pointer-Shift and Prefuzz share accessors, and to the inlined hash function we introduced an $\mathit{if} s = 0$ test that dispatches to Prefuzz.
This was intended to reduce pressure on the instruction cache, which is an important consideration in macrobenchmarks.
To minimize the performance penalty on rehashing, we lift this dispatch out of the rehashing loop in the same way that dispatches to different accessors are.

In \lisp{equal} hash tables, the truncation limit and the current max-chain-length are packed into a single machine integer and stored in the same slot that we used for $s$ in the \lisp{eq} case.


\section{Benchmarking Environment}
\label{sec:benchmarking-environment}

All reported results are from a single Intel Core i7-1185G7 laptop running Linux with the \texttt{performance} scaling governor.
Turbo boost and CPU idle states were disabled (\texttt{cpupower -c 3 idle-set -D 0}).
The benchmarking process was run at maximum priority (\texttt{nice -n -20}), pinned to a single CPU (\texttt{taskset -c 3}).


\section{Microbenchmarking Methodology}
\label{sec:microbenchmarking-methodology}

We plot the estimated time it takes to do a particular operation vs the number of keys (between $1$ and at most $2^{24}$) in the hash table for different key types and allocation patterns.
All hash table variants compared in this paper allocate memory only when the insertion of a new key requires growing the internal data structures.\footnote{Moreover, they do so in an identical and deterministic pattern, so we exclude the time spent in garbage collection from the measurements.
The exception to this pattern is the Constant hash, whose specialized implementation has a smaller than usual memory usage.}
To reduce the computational burden of the benchmarking process, we list ranges of key counts within which no resizing takes place and only measure performance at the minimum and maximum key count in each range, assuming that linear interpolation is a reasonable approximation in between.

At each such key count, we then estimate the average time it takes to perform a hash table operation (e.g. PUT).
First, a set of keys is generated for the given type (e.g. \texttt{FIXNUM}) and allocation pattern (e.g. \texttt{:RND 6}).
The hash table is allocated in an empty state with comparison function \lisp{equal} for string keys or \lisp{eq} in all other cases.
Then, we measure the average time it takes to perform a given operation for keys in the key set:
\begin{itemize}
\item{PUT}: Inserting a key when populating the hash table.
Keys are inserted in the order they were generated.
\item{GET}: Looking up a key in the key set.
Random order.
\item{MISS}: Looking up a key not in the key set.
Random order.
\item{DEL}: Deleting a key in the key set.
Random order.
\end{itemize}
These steps are performed in the order listed.
That is, first the hash table is populated, and PUT is measured, followed by GET then MISS.
Finally, DEL timings are taken.
At the end of the DEL phase, the hash table is once again empty.

These average times over key sets have a low relative standard deviation of about 0\%--2\%.
To reduce the variance further, we take multiple such measurements and report their average.
In particular, we take at least 3 measurements, then continue until the total number operations performed exceeds \num{5000000}.

\begin{table}
  \caption{Estimated means and relative standard errors of real (wall clock) and CPU times in seconds to \emph{compile and load} a set of libraries.}
  \label{tab:load-system-results}
  \centering
  \begin{tabular}{@{}lrrrr@{}}
                   & Real    & $\pm$RSE\% &     CPU & $\pm$RSE\% \\
\midrule
Pr                 & 24.068  &     0.02\% &  24.037 &     0.02\% \\
Mu                 & 24.152  &     0.01\% &  24.117 &     0.01\% \\
Co+Pr              & 23.976  &     0.03\% &  23.945 &     0.02\% \\
Co+Mu              & 23.979  &     0.03\% &  23.943 &     0.03\% \\
Co+Pr>Mu           & 23.988  &     0.02\% &  23.955 &     0.02\% \\
Co+PS>Pr>Mu        & 23.951  &     0.02\% &  23.918 &     0.01\% \\
Equal*             & 23.824  &     0.02\% &  23.792 &     0.02\% \\
\midrule
  \end{tabular}
\end{table}

\begin{table}
  \caption{Estimated times to \emph{test} a set of libraries.}
  \label{tab:test-system-results}
  \centering
  \begin{tabular}{@{}lrrrr@{}}
                   & Real    & $\pm$RSE\% & CPU     & $\pm$RSE\% \\
\midrule
Pr                 & 25.512  &     0.03\% &  24.493 &     0.02\% \\
Mu                 & 25.632  &     0.05\% &  24.632 &     0.04\% \\
Co+Pr              & 25.367  &     0.03\% &  24.352 &     0.02\% \\
Co+Mu              & 25.360  &     0.04\% &  24.342 &     0.03\% \\
Co+Pr>Mu           & 25.385  &     0.04\% &  24.367 &     0.03\% \\
Co+PS>Pr>Mu        & 25.372  &     0.03\% &  24.374 &     0.02\% \\
Equal*             & 25.257  &     0.03\% &  24.256 &     0.02\% \\
\midrule
  \end{tabular}
\end{table}

Finally, with fewer than 100 keys, timing granularity is a limiting factor, so we allocate a number of hash tables (plus the two key sets for each, the second being for MISS) and measure the total time it takes to e.g. populate them.
The number of hash tables is chosen such that the total number of keys is at least 100.
This is a low enough number that the total memory footprint of the allocated hash tables stays below 32KB, the CPU's L1 cache size in our benchmarking environment.
SBCL had a 16GB heap.


\section{Macrobenchmark Results}
\label{sec:macrobenchmark-results}

Here, we describe our macrobenchmarking experiments in \Cref{sec:macrobenchmarks} in more detail.
In the first suite, 16 libaries were compiled and loaded with \lisp{(asdf:load-system <library> :force t)}.
In the second, the tests of the same libraries were run with \lisp{(asdf:test-system <library>)}.
In the third, SBCL tests were run with file-by-file with \texttt{tests/run-tests.sh}.

\begin{table}
  \caption{Estimated times to \emph{run SBCL tests}.}
  \label{tab:sbcl-tests-results}
  \centering
  \begin{tabular}{@{}lrrrr@{}}
                   & Real    & $\pm$RSE\% & CPU     & $\pm$RSE\% \\
\midrule
Pr                 & 582.16 &     0.02\% & 452.34 &  0.03\% \\
Mu                 & 582.85 &     0.02\% & 453.06 &  0.02\% \\
Co+Pr              & 580.07 &     0.02\% & 450.29 &  0.03\% \\
Co+Mu              & 579.44 &     0.02\% & 449.64 &  0.03\% \\
Co+Pr>Mu           & 578.90 &     0.02\% & 449.03 &  0.02\% \\
Co+PS>Pr>Mu        & 578.20 &     0.02\% & 448.39 &  0.03\% \\
Equal*             & 568.89 &     0.02\% & 439.09 &  0.02\% \\
\midrule
  \end{tabular}
\end{table}

We compared the following configurations (following a naming convention like in \Cref{alg:rehash-eq}):
\begin{enumerate}
\item \textbf{Pr}: unchanged baseline version with the Prefuzz \lisp{eq} hash and the non-adaptive \lisp{equal} hash
\item \textbf{Mu}: the \lisp{eq} hash was changed to Murmur3
\item \textbf{Co+Pr}: Constant hash followed by Prefuzz
\item \textbf{Co+Mu}: Constant hash followed by Murmur3
\item \textbf{Co+Pr>Mu}: Constant hash followed by Prefuzz with fallback to Murmur3
\item \textbf{Co+PS>Pr>Mu}: Constant hash followed by Pointer-Shift with fallback to Prefuzz then to Murmur3 (called Adaptive in \Cref{sec:microbenchmarks}).
\item \textbf{Equal*}: Like the previous, but the \lisp{equal} hash is also adaptive (see \Cref{sec:half-an-example}).
\end{enumerate}
Note that since Pr and Co+Pr lack an eventual fallback to Murmur, they have an easy-to-trigger unbounded worst case, so it may be pertinent to think of Mu as the baseline.

In the benchmarking environment described in \Cref{sec:benchmarking-environment}, all individual benchmarks in the three benchmark suites (e.g. running the tests of a single library or a single SBCL test file) were run 10 times on each configuration with their runs interleaved to reduce the effect of correlated noise.
We estimated the mean time a benchmark took on each configuration and computed the standard error of this estimated mean.
We report the total estimated mean times for each benchmark suite along with their relative standard errors (RSE) in \Cref{tab:load-system-results,tab:test-system-results,tab:sbcl-tests-results}.


\raggedbottom

\section{Results for Strings}
\label{sec:results-for-strings}

Keys are sampled from the set of all (about \num{40000}) strings in the running Lisp.
This includes names of symbols, packages, docstrings, etc.
The 10th and 90th percentiles of the distribution of the string lengths is 7 and 39.
The MISS key set consists of random strings with length sampled uniformly from the $[4,44]$ interval.

\begin{figure}[H]
\figstringexistingregret
\label{fig:string-existing-regret-2}
\end{figure}

\begin{figure}[H]
\figstringexistingput
\label{fig:string-existing-put-2}
\end{figure}

\begin{figure}[H]
\figstringexistingget
\label{fig:string-existing-get-2}
\end{figure}

\begin{figure}[H]
\begin{tikzpicture}
  \begin{semilogxaxis}[ylabel=ns / miss,
      xmin=1, xmode=log, log basis x={2},
      ymin=30, ymax=200, ymode=log, log basis y={2},
      legend pos=north west,
      legend style={nodes={scale=0.7, transform shape}},
      legend cell align={left},
      height=0.6*\linewidth,
      width=0.98\linewidth,
    ]
    \pgfplotstableread{data/string-existing-sbcl.tbl}{\sorted}
    \addplot[uniformhash] table [x=nkeys, y=missns] {\sorted};
    \addlegendentry{SBCL};

    \pgfplotstableread{data/string-existing-adaptive.tbl}{\sorted}
    \addplot[adaptivehash] table [x=nkeys, y=missns] {\sorted};
    \addlegendentry{\textbf{Adaptive}};
  \end{semilogxaxis}
\end{tikzpicture}
\caption{MISS timings with string keys}
\label{fig:string-existing-miss}
\end{figure}

\begin{figure}[H]
\begin{tikzpicture}
  \begin{semilogxaxis}[ylabel=ns / del,
      xmin=1, xmode=log, log basis x={2},
      ymin=50, ymax=200, ymode=log, log basis y={2},
      legend pos=north west,
      legend style={nodes={scale=0.7, transform shape}},
      legend cell align={left},
      height=0.6*\linewidth,
      width=0.98\linewidth,
    ]
    \pgfplotstableread{data/string-existing-sbcl.tbl}{\sorted}
    \addplot[uniformhash] table [x=nkeys, y=delns] {\sorted};
    \addlegendentry{SBCL};

    \pgfplotstableread{data/string-existing-adaptive.tbl}{\sorted}
    \addplot[adaptivehash] table [x=nkeys, y=delns] {\sorted};
    \addlegendentry{\textbf{Adaptive}};
  \end{semilogxaxis}
\end{tikzpicture}
\caption{DEL timings with string keys}
\label{fig:string-existing-del}
\end{figure}


\clearpage
\section{Results for \texttt{FIXNUM} \texttt{:PROG 1}}
\label{sec:results-for-fixnum-prog-1}

Keys form  an arithmetic progression with difference \texttt{1} starting from a large random offset and are used in that order for PUT.
We also generate a set of keys not in the hash table for MISS, by using another, suitable offset (so that the two sets are disjunct).
For GET, MISS, and DEL, keys are presented in random order.

\begin{figure}[H]
\figfixnumprogoneregret
\label{fig:fixnum-prog-1-regret-2}
\end{figure}

\begin{figure}[H]
\figfixnumprogoneput
\label{fig:fixnum-prog-1-put-2}
\end{figure}

\begin{figure}[H]
\figfixnumprogoneget
\label{fig:fixnum-prog-1-get-2}
\end{figure}

\begin{figure}[H]
\tikzpicturedependsonfile{data/fixnum-prog-1-murmur}
\tikzpicturedependsonfile{data/fixnum-prog-1-prefuzz}
\tikzpicturedependsonfile{data/fixnum-prog-1-flat}
\tikzpicturedependsonfile{data/fixnum-prog-1-flat-safe-small}
\begin{tikzpicture}
  \begin{axis}[ylabel=ns / miss,
      xmin=1, xmode=log, log basis x={2},
      ymin=6, ymax=180, ymode=log, log basis y={2},
      legend pos=north west,
      legend style={nodes={scale=0.7, transform shape}},
      legend cell align={left},
      height=0.6*\linewidth,
      width=0.98\linewidth,
    ]
    \pgfplotstableread{data/fixnum-prog-1-murmur.tbl}{\sorted}
    \addplot[murmurhash] table [x=nkeys, y=missns] {\sorted};
    \addlegendentry{Murmur};

    \pgfplotstableread{data/fixnum-prog-1-prefuzz.tbl}{\sorted}
    \addplot[prefuzzhash] table [x=nkeys, y=missns] {\sorted};
    \addlegendentry{Prefuzz};

    \pgfplotstableread{data/fixnum-prog-1-flat.tbl}{\sorted}
    \addplot[constanthash] table [x=nkeys, y=missns] {\sorted};
    \addlegendentry{Co+Pr};

    \pgfplotstableread{data/fixnum-prog-1-flat-safe-small.tbl}{\sorted}
    \addplot[adaptivehash] table [x=nkeys, y=missns] {\sorted};
    \addlegendentry{\textbf{Adaptive}};
  \end{axis}
\end{tikzpicture}
\caption{MISS timings with \texttt{FIXNUM} \texttt{:PROG 1}}
\label{fig:fixnum-prog-1-miss}
\end{figure}

\begin{figure}[H]
\tikzpicturedependsonfile{data/fixnum-prog-1-murmur}
\tikzpicturedependsonfile{data/fixnum-prog-1-prefuzz}
\tikzpicturedependsonfile{data/fixnum-prog-1-flat}
\tikzpicturedependsonfile{data/fixnum-prog-1-flat-safe-small}
\begin{tikzpicture}
  \begin{axis}[ylabel=ns / del,
      xmin=1, xmode=log, log basis x={2},
      ymin=6, ymax=180, ymode=log, log basis y={2},
      legend pos=north west,
      legend style={nodes={scale=0.7, transform shape}},
      legend cell align={left},
      height=0.6*\linewidth,
      width=0.98\linewidth,
    ]
    \tikzpicturedependsonfile{data/fixnum-prog-1-murmur}
    \tikzpicturedependsonfile{data/fixnum-prog-1-prefuzz}
    \tikzpicturedependsonfile{data/fixnum-prog-1-flat}
    \tikzpicturedependsonfile{data/fixnum-prog-1-flat-safe-small}

    \pgfplotstableread{data/fixnum-prog-1-murmur.tbl}{\sorted}
    \addplot[murmurhash] table [x=nkeys, y=delns] {\sorted};
    \addlegendentry{Murmur};

    \pgfplotstableread{data/fixnum-prog-1-prefuzz.tbl}{\sorted}
    \addplot[prefuzzhash] table [x=nkeys, y=delns] {\sorted};
    \addlegendentry{Prefuzz};

    \pgfplotstableread{data/fixnum-prog-1-flat.tbl}{\sorted}
    \addplot[constanthash] table [x=nkeys, y=delns] {\sorted};
    \addlegendentry{Co+Pr};

    \pgfplotstableread{data/fixnum-prog-1-flat-safe-small.tbl}{\sorted}
    \addplot[adaptivehash] table [x=nkeys, y=delns] {\sorted};
    \addlegendentry{\textbf{Adaptive}};
  \end{axis}
\end{tikzpicture}
\caption{DEL timings with \texttt{FIXNUM} \texttt{:PROG 1}}
\label{fig:fixnum-prog-1-del}
\end{figure}


\clearpage
\section{Results for \texttt{FLOAT} \texttt{:PROG 1}}
\label{sec:results-for-float-prog-1}

Keys are generated as in \Cref{sec:results-for-fixnum-prog-1}, but the \lisp{fixnum} values are converted to \lisp{single-float}.

\begin{figure}[H]
\figfloatprogoneregret
\label{fig:float-prog-1-regret-2}
\end{figure}

\begin{figure}[H]
\begin{tikzpicture}
  \begin{axis}[ylabel=ns / put,
      xmin=1, xmode=log, log basis x={2},
      ymin=20, ymax=2000, ymode=log, log basis y={2},
      legend pos=north east,
      legend style={nodes={scale=0.7, transform shape}},
      legend cell align={left},
      height=0.6*\linewidth,
      width=0.98\linewidth,
    ]
    \pgfplotstableread{data/float-prog-1-murmur.tbl}{\sorted}
    \addplot[murmurhash] table [x=nkeys, y=putns] {\sorted};
    \addlegendentry{Murmur};

    \pgfplotstableread{data/float-prog-1-prefuzz.tbl}{\sorted}
    \addplot[prefuzzhash] table [x=nkeys, y=putns] {\sorted};
    \addlegendentry{Prefuzz};

    \pgfplotstableread{data/float-prog-1-flat.tbl}{\sorted}
    \addplot[constanthash] table [x=nkeys, y=putns] {\sorted};
    \addlegendentry{Co+Pr};

    \pgfplotstableread{data/float-prog-1-flat-safe-small.tbl}{\sorted}
    \addplot[adaptivehash] table [x=nkeys, y=putns] {\sorted};
    \addlegendentry{\textbf{Adaptive}};
  \end{axis}
\end{tikzpicture}
\caption{PUT timings with \texttt{FLOAT} \texttt{:PROG 1}}
\label{fig:float-prog-1-put}
\end{figure}

\begin{figure}[H]
\begin{tikzpicture}
  \begin{axis}[ylabel=ns / get,
      xmin=1, xmode=log, log basis x={2},
      ymin=6, ymax=2000, ymode=log, log basis y={2},
      legend pos=north east,
      legend style={nodes={scale=0.7, transform shape}},
      legend cell align={left},
      height=0.6*\linewidth,
      width=0.98\linewidth,
    ]
    \pgfplotstableread{data/float-prog-1-murmur.tbl}{\sorted}
    \addplot[murmurhash] table [x=nkeys, y=getns] {\sorted};
    \addlegendentry{Murmur};

    \pgfplotstableread{data/float-prog-1-prefuzz.tbl}{\sorted}
    \addplot[prefuzzhash] table [x=nkeys, y=getns] {\sorted};
    \addlegendentry{Prefuzz};

    \pgfplotstableread{data/float-prog-1-flat.tbl}{\sorted}
    \addplot[constanthash] table [x=nkeys, y=getns] {\sorted};
    \addlegendentry{Co+Pr};

    \pgfplotstableread{data/float-prog-1-flat-safe-small.tbl}{\sorted}
    \addplot[adaptivehash] table [x=nkeys, y=getns] {\sorted};
    \addlegendentry{\textbf{Adaptive}};
  \end{axis}
\end{tikzpicture}
\caption{GET timings with \texttt{FLOAT} \texttt{:PROG 1}}
\label{fig:float-prog-1-get}
\end{figure}

\begin{figure}[H]
\tikzpicturedependsonfile{data/float-prog-1-murmur}
\tikzpicturedependsonfile{data/float-prog-1-prefuzz}
\tikzpicturedependsonfile{data/float-prog-1-flat}
\tikzpicturedependsonfile{data/float-prog-1-flat-safe-small}
\begin{tikzpicture}
  \begin{axis}[ylabel=ns / miss,
      xmin=1, xmode=log, log basis x={2},
      ymin=6, ymax=2000, ymode=log, log basis y={2},
      legend pos=north east,
      legend style={nodes={scale=0.7, transform shape}},
      legend cell align={left},
      height=0.6*\linewidth,
      width=0.98\linewidth,
    ]
    \pgfplotstableread{data/float-prog-1-murmur.tbl}{\sorted}
    \addplot[murmurhash] table [x=nkeys, y=missns] {\sorted};
    \addlegendentry{Murmur};

    \pgfplotstableread{data/float-prog-1-prefuzz.tbl}{\sorted}
    \addplot[prefuzzhash] table [x=nkeys, y=missns] {\sorted};
    \addlegendentry{Prefuzz};

    \pgfplotstableread{data/float-prog-1-flat.tbl}{\sorted}
    \addplot[constanthash] table [x=nkeys, y=missns] {\sorted};
    \addlegendentry{Co+Pr};

    \pgfplotstableread{data/float-prog-1-flat-safe-small.tbl}{\sorted}
    \addplot[adaptivehash] table [x=nkeys, y=missns] {\sorted};
    \addlegendentry{\textbf{Adaptive}};
  \end{axis}
\end{tikzpicture}
\caption{MISS timings with \texttt{FLOAT} \texttt{:PROG 1}}
\label{fig:float-prog-1-miss}
\end{figure}

\begin{figure}[H]
\tikzpicturedependsonfile{data/float-prog-1-murmur}
\tikzpicturedependsonfile{data/float-prog-1-prefuzz}
\tikzpicturedependsonfile{data/float-prog-1-flat}
\tikzpicturedependsonfile{data/float-prog-1-flat-safe-small}
\begin{tikzpicture}
  \begin{axis}[ylabel=ns / del,
      xmin=1, xmode=log, log basis x={2},
      ymin=6, ymax=2000, ymode=log, log basis y={2},
      legend pos=north east,
      legend style={nodes={scale=0.7, transform shape}},
      legend cell align={left},
      height=0.6*\linewidth,
      width=0.98\linewidth,
    ]
    \tikzpicturedependsonfile{data/float-prog-1-murmur}
    \tikzpicturedependsonfile{data/float-prog-1-prefuzz}
    \tikzpicturedependsonfile{data/float-prog-1-flat}
    \tikzpicturedependsonfile{data/float-prog-1-flat-safe-small}

    \pgfplotstableread{data/float-prog-1-murmur.tbl}{\sorted}
    \addplot[murmurhash] table [x=nkeys, y=delns] {\sorted};
    \addlegendentry{Murmur};

    \pgfplotstableread{data/float-prog-1-prefuzz.tbl}{\sorted}
    \addplot[prefuzzhash] table [x=nkeys, y=delns] {\sorted};
    \addlegendentry{Prefuzz};

    \pgfplotstableread{data/float-prog-1-flat.tbl}{\sorted}
    \addplot[constanthash] table [x=nkeys, y=delns] {\sorted};
    \addlegendentry{Co+Pr};

    \pgfplotstableread{data/float-prog-1-flat-safe-small.tbl}{\sorted}
    \addplot[adaptivehash] table [x=nkeys, y=delns] {\sorted};
    \addlegendentry{\textbf{Adaptive}};
  \end{axis}
\end{tikzpicture}
\caption{DEL timings with \texttt{FLOAT} \texttt{:PROG 1}}
\label{fig:float-prog-1-del}
\end{figure}


\clearpage
\section{Results for \texttt{FIXNUM} \texttt{:PROG 12}}
\label{sec:results-for-fixnum-prog-12}

Same as \texttt{FIXNUM} \texttt{:PROG 1}, but with a difference of 12.

\begin{figure}[H]
\figfixnumprogtwelveregret
\label{fig:fixnum-prog-12-regret-2}
\end{figure}

\begin{figure}[H]
\begin{tikzpicture}
  \begin{axis}[ylabel=ns / put,
      xmin=1, xmode=log, log basis x={2},
      ymin=20, ymax=180, ymode=log, log basis y={2},
      legend pos=north west,
      legend style={nodes={scale=0.7, transform shape}},
      legend cell align={left},
      height=0.6*\linewidth,
      width=0.98\linewidth,
    ]
    \pgfplotstableread{data/fixnum-prog-12-murmur.tbl}{\sorted}
    \addplot[murmurhash] table [x=nkeys, y=putns] {\sorted};
    \addlegendentry{Murmur};

    \pgfplotstableread{data/fixnum-prog-12-prefuzz.tbl}{\sorted}
    \addplot[prefuzzhash] table [x=nkeys, y=putns] {\sorted};
    \addlegendentry{Prefuzz};

    \pgfplotstableread{data/fixnum-prog-12-flat.tbl}{\sorted}
    \addplot[constanthash] table [x=nkeys, y=putns] {\sorted};
    \addlegendentry{Co+Pr};

    \pgfplotstableread{data/fixnum-prog-12-flat-safe-small.tbl}{\sorted}
    \addplot[adaptivehash] table [x=nkeys, y=putns] {\sorted};
    \addlegendentry{\textbf{Adaptive}};
  \end{axis}
\end{tikzpicture}
\caption{PUT timings with \texttt{FIXNUM} \texttt{:PROG 12}.
Prefuzz outperforms Murmur due to its speed, lower regret and cache-friendliness.
Adaptive is able to benefit from its advantage in regret only at larger sizes.}
\label{fig:fixnum-prog-12-put}
\end{figure}

\begin{figure}[H]
\begin{tikzpicture}
  \begin{axis}[ylabel=ns / get,
      xmin=1, xmode=log, log basis x={2},
      ymin=6, ymax=180, ymode=log, log basis y={2},
      legend pos=north west,
      legend style={nodes={scale=0.7, transform shape}},
      legend cell align={left},
      height=0.6*\linewidth,
      width=0.98\linewidth,
    ]
    \pgfplotstableread{data/fixnum-prog-12-murmur.tbl}{\sorted}
    \addplot[murmurhash] table [x=nkeys, y=getns] {\sorted};
    \addlegendentry{Murmur};

    \pgfplotstableread{data/fixnum-prog-12-prefuzz.tbl}{\sorted}
    \addplot[prefuzzhash] table [x=nkeys, y=getns] {\sorted};
    \addlegendentry{Prefuzz};

    \pgfplotstableread{data/fixnum-prog-12-flat.tbl}{\sorted}
    \addplot[constanthash] table [x=nkeys, y=getns] {\sorted};
    \addlegendentry{Co+Pr};

    \pgfplotstableread{data/fixnum-prog-12-flat-safe-small.tbl}{\sorted}
    \addplot[adaptivehash] table [x=nkeys, y=getns] {\sorted};
    \addlegendentry{\textbf{Adaptive}};
  \end{axis}
\end{tikzpicture}
\caption{GET timings with \texttt{FIXNUM} \texttt{:PROG 12}.
Compared to PUT, differences in regret translate more clearly to lookup performance because keys are queried in random order.}
\label{fig:fixnum-prog-12-get}
\end{figure}

\begin{figure}[H]
\begin{tikzpicture}
  \begin{axis}[ylabel=ns / miss,
      xmin=1, xmode=log, log basis x={2},
      ymin=6, ymax=180, ymode=log, log basis y={2},
      legend pos=north west,
      legend style={nodes={scale=0.7, transform shape}},
      legend cell align={left},
      height=0.6*\linewidth,
      width=0.98\linewidth,
    ]
    \pgfplotstableread{data/fixnum-prog-12-murmur.tbl}{\sorted}
    \addplot[murmurhash] table [x=nkeys, y=missns] {\sorted};
    \addlegendentry{Murmur};

    \pgfplotstableread{data/fixnum-prog-12-prefuzz.tbl}{\sorted}
    \addplot[prefuzzhash] table [x=nkeys, y=missns] {\sorted};
    \addlegendentry{Prefuzz};

    \pgfplotstableread{data/fixnum-prog-12-flat.tbl}{\sorted}
    \addplot[constanthash] table [x=nkeys, y=missns] {\sorted};
    \addlegendentry{Co+Pr};

    \pgfplotstableread{data/fixnum-prog-12-flat-safe-small.tbl}{\sorted}
    \addplot[adaptivehash] table [x=nkeys, y=missns] {\sorted};
    \addlegendentry{\textbf{Adaptive}};
  \end{axis}
\end{tikzpicture}
\caption{MISS timings with \texttt{FIXNUM} \texttt{:PROG 12}}
\label{fig:fixnum-prog-12-miss}
\end{figure}

\begin{figure}[H]
\begin{tikzpicture}
  \begin{axis}[ylabel=ns / del,
      xmin=1, xmode=log, log basis x={2},
      ymin=6, ymax=180, ymode=log, log basis y={2},
      legend pos=north west,
      legend style={nodes={scale=0.7, transform shape}},
      legend cell align={left},
      height=0.6*\linewidth,
      width=0.98\linewidth,
    ]
    \pgfplotstableread{data/fixnum-prog-12-murmur.tbl}{\sorted}
    \addplot[murmurhash] table [x=nkeys, y=delns] {\sorted};
    \addlegendentry{Murmur};

    \pgfplotstableread{data/fixnum-prog-12-prefuzz.tbl}{\sorted}
    \addplot[prefuzzhash] table [x=nkeys, y=delns] {\sorted};
    \addlegendentry{Prefuzz};

    \pgfplotstableread{data/fixnum-prog-12-flat.tbl}{\sorted}
    \addplot[constanthash] table [x=nkeys, y=delns] {\sorted};
    \addlegendentry{Co+Pr};

    \pgfplotstableread{data/fixnum-prog-12-flat-safe-small.tbl}{\sorted}
    \addplot[adaptivehash] table [x=nkeys, y=delns] {\sorted};
    \addlegendentry{\textbf{Adaptive}};
  \end{axis}
\end{tikzpicture}
\caption{DEL timings with \texttt{FIXNUM} \texttt{:PROG 12}}
\label{fig:fixnum-prog-12-del}
\end{figure}


\clearpage
\section{Results for \texttt{FIXNUM} \texttt{:RND 6}}
\label{sec:results-for-fixnum-rnd-6}

Similar to \texttt{:PROG 6}, but the difference between successive keys is sampled uniformly from the $[0, 5]$ interval.

\begin{figure}[H]
\begin{tikzpicture}
  \begin{axis}[xlabel=number of keys, ylabel=regret,
      xmin=1, xmode=log, log basis x={2},
      ymin=-0.01, ymax=0.79,
      legend pos=north west,
      legend columns=2,
      legend style={nodes={scale=0.7, transform shape}},
      legend cell align={left},
      height=0.6*\linewidth,
      width=0.98\linewidth,
    ]
    \pgfplotstableread{data/fixnum-rnd-6-murmur.tbl}{\sorted}
    \addplot[uniformhash] table [x=nkeys, y=rndregret] {\sorted};
    \addlegendentry{Uniform};

    \pgfplotstableread{data/fixnum-rnd-6-murmur.tbl}{\sorted}
    \addplot[murmurhash] table [x=nkeys, y=regret] {\sorted};
    \addlegendentry{Murmur};

    \pgfplotstableread{data/fixnum-rnd-6-prefuzz.tbl}{\sorted}
    \addplot[prefuzzhash] table [x=nkeys, y=regret] {\sorted};
    \addlegendentry{Prefuzz};

    \pgfplotstableread{data/fixnum-rnd-6-flat.tbl}{\sorted}
    \addplot[constanthash] table [x=nkeys, y=regret] {\sorted};
    \addlegendentry{Co+Pr};

    \pgfplotstableread{data/fixnum-rnd-6-flat-safe-small.tbl}{\sorted}
    \addplot[adaptivehash] table [x=nkeys, y=regret] {\sorted};
    \addlegendentry{\textbf{Adaptive}};
    \addplot[mark=*,mark options={scale=1.0}] coordinates {(33,0)};
  \end{axis}
\end{tikzpicture}
\caption{Regret with \texttt{FIXNUM} \texttt{:RND 6}.
Prefuzz does better than Murmur initially but gradually loses its edge.
Adaptive keeps its edge.}
\label{fig:fixnum-rnd-6-regret}
\end{figure}

\begin{figure}[H]
\begin{tikzpicture}
  \begin{axis}[ylabel=ns / put,
      xmin=1, xmode=log, log basis x={2},
      ymin=20, ymax=180, ymode=log, log basis y={2},
      legend pos=north west,
      legend style={nodes={scale=0.7, transform shape}},
      legend cell align={left},
      height=0.6*\linewidth,
      width=0.98\linewidth,
    ]
    \pgfplotstableread{data/fixnum-rnd-6-murmur.tbl}{\sorted}
    \addplot[murmurhash] table [x=nkeys, y=putns] {\sorted};
    \addlegendentry{Murmur};

    \pgfplotstableread{data/fixnum-rnd-6-prefuzz.tbl}{\sorted}
    \addplot[prefuzzhash] table [x=nkeys, y=putns] {\sorted};
    \addlegendentry{Prefuzz};

    \pgfplotstableread{data/fixnum-rnd-6-flat.tbl}{\sorted}
    \addplot[constanthash] table [x=nkeys, y=putns] {\sorted};
    \addlegendentry{Co+Pr};

    \pgfplotstableread{data/fixnum-rnd-6-flat-safe-small.tbl}{\sorted}
    \addplot[adaptivehash] table [x=nkeys, y=putns] {\sorted};
    \addlegendentry{\textbf{Adaptive}};
  \end{axis}
\end{tikzpicture}
\caption{PUT timings with \texttt{FIXNUM} \texttt{:RND 6}. Once again, the advantage of Prefuzz over Murmur grows with size because it is friendlier to the cache. Adaptive manages to translate some of its lead in regret into outright speed.}
\label{fig:fixnum-rnd-6-put}
\end{figure}

\begin{figure}[H]
\begin{tikzpicture}
  \begin{axis}[ylabel=ns / get,
      xmin=1, xmode=log, log basis x={2},
      ymin=6, ymax=180, ymode=log, log basis y={2},
      legend pos=north west,
      legend style={nodes={scale=0.7, transform shape}},
      legend cell align={left},
      height=0.6*\linewidth,
      width=0.98\linewidth,
    ]
    \pgfplotstableread{data/fixnum-rnd-6-murmur.tbl}{\sorted}
    \addplot[murmurhash] table [x=nkeys, y=getns] {\sorted};
    \addlegendentry{Murmur};

    \pgfplotstableread{data/fixnum-rnd-6-prefuzz.tbl}{\sorted}
    \addplot[prefuzzhash] table [x=nkeys, y=getns] {\sorted};
    \addlegendentry{Prefuzz};

    \pgfplotstableread{data/fixnum-rnd-6-flat.tbl}{\sorted}
    \addplot[constanthash] table [x=nkeys, y=getns] {\sorted};
    \addlegendentry{Co+Pr};

    \pgfplotstableread{data/fixnum-rnd-6-flat-safe-small.tbl}{\sorted}
    \addplot[adaptivehash] table [x=nkeys, y=getns] {\sorted};
    \addlegendentry{\textbf{Adaptive}};
  \end{axis}
\end{tikzpicture}
\caption{GET timings with \texttt{FIXNUM} \texttt{:RND 6}. Prefuzz is considerably ahead of Murmur at all sizes. Adaptive is better than Prefuzz at larger sizes.}
\label{fig:fixnum-rnd-6-get}
\end{figure}

\begin{figure}[H]
\begin{tikzpicture}
  \begin{axis}[ylabel=ns / miss,
      xmin=1, xmode=log, log basis x={2},
      ymin=6, ymax=180, ymode=log, log basis y={2},
      legend pos=north west,
      legend style={nodes={scale=0.7, transform shape}},
      legend cell align={left},
      height=0.6*\linewidth,
      width=0.98\linewidth,
    ]
    \pgfplotstableread{data/fixnum-rnd-6-murmur.tbl}{\sorted}
    \addplot[murmurhash] table [x=nkeys, y=missns] {\sorted};
    \addlegendentry{Murmur};

    \pgfplotstableread{data/fixnum-rnd-6-prefuzz.tbl}{\sorted}
    \addplot[prefuzzhash] table [x=nkeys, y=missns] {\sorted};
    \addlegendentry{Prefuzz};

    \pgfplotstableread{data/fixnum-rnd-6-flat.tbl}{\sorted}
    \addplot[constanthash] table [x=nkeys, y=missns] {\sorted};
    \addlegendentry{Co+Pr};

    \pgfplotstableread{data/fixnum-rnd-6-flat-safe-small.tbl}{\sorted}
    \addplot[adaptivehash] table [x=nkeys, y=missns] {\sorted};
    \addlegendentry{\textbf{Adaptive}};
  \end{axis}
\end{tikzpicture}
\caption{MISS timings with \texttt{FIXNUM} \texttt{:RND 6}}
\label{fig:fixnum-rnd-6-miss}
\end{figure}

\begin{figure}[H]
\begin{tikzpicture}
  \begin{axis}[ylabel=ns / del,
      xmin=1, xmode=log, log basis x={2},
      ymin=6, ymax=180, ymode=log, log basis y={2},
      legend pos=north west,
      legend style={nodes={scale=0.7, transform shape}},
      legend cell align={left},
      height=0.6*\linewidth,
      width=0.98\linewidth,
    ]
    \pgfplotstableread{data/fixnum-rnd-6-murmur.tbl}{\sorted}
    \addplot[murmurhash] table [x=nkeys, y=delns] {\sorted};
    \addlegendentry{Murmur};

    \pgfplotstableread{data/fixnum-rnd-6-prefuzz.tbl}{\sorted}
    \addplot[prefuzzhash] table [x=nkeys, y=delns] {\sorted};
    \addlegendentry{Prefuzz};

    \pgfplotstableread{data/fixnum-rnd-6-flat.tbl}{\sorted}
    \addplot[constanthash] table [x=nkeys, y=delns] {\sorted};
    \addlegendentry{Co+Pr};

    \pgfplotstableread{data/fixnum-rnd-6-flat-safe-small.tbl}{\sorted}
    \addplot[adaptivehash] table [x=nkeys, y=delns] {\sorted};
    \addlegendentry{\textbf{Adaptive}};
  \end{axis}
\end{tikzpicture}
\caption{DEL timings with \texttt{FIXNUM} \texttt{:RND 6}}
\label{fig:fixnum-rnd-6-del}
\end{figure}


\clearpage
\section{Results for \texttt{CONS} \texttt{:RND 6}}
\label{sec:results-for-cons-rnd-6}

Like the \texttt{FIXNUM} \texttt{:RND 6}, but keys are \lisp{cons} objects with a random number of conses allocated between them.
There is no explicit random offset here; we rely on the addresses assigned by the memory allocator.
To prevent the garbage collector from compacting memory regions holding these objects in memory, the skipped over conses are kept alive.

In all experiments, there is some unexplained weirdness at large sizes, which affects all hashes except Murmur.
Even the regret of Adaptive is affected: it improves to an unexpected level but does so very erratically.

\begin{figure}[H]
\begin{tikzpicture}
  \begin{axis}[xlabel=number of keys, ylabel=regret,
      xmin=1, xmode=log, log basis x={2},
      ymin=-0.01, ymax=0.79,
      legend pos=north west,
      legend columns=2,
      legend style={nodes={scale=0.7, transform shape}},
      legend cell align={left},
      height=0.6*\linewidth,
      width=0.98\linewidth,
    ]
    \pgfplotstableread{data/cons-rnd-6-murmur.tbl}{\sorted}
    \addplot[uniformhash] table [x=nkeys, y=rndregret] {\sorted};
    \addlegendentry{Uniform};

    \pgfplotstableread{data/cons-rnd-6-murmur.tbl}{\sorted}
    \addplot[murmurhash] table [x=nkeys, y=regret] {\sorted};
    \addlegendentry{Murmur};

    \pgfplotstableread{data/cons-rnd-6-prefuzz.tbl}{\sorted}
    \addplot[prefuzzhash] table [x=nkeys, y=regret] {\sorted};
    \addlegendentry{Prefuzz};

    \pgfplotstableread{data/cons-rnd-6-flat.tbl}{\sorted}
    \addplot[constanthash] table [x=nkeys, y=regret] {\sorted};
    \addlegendentry{Co+Pr};

    \pgfplotstableread{data/cons-rnd-6-flat-safe-small.tbl}{\sorted}
    \addplot[adaptivehash] table [x=nkeys, y=regret] {\sorted};
    \addlegendentry{\textbf{Adaptive}};
    \addplot[mark=*,mark options={scale=1.0}] coordinates {(33,0)};
  \end{axis}
\end{tikzpicture}
\caption{Regret with \texttt{CONS} \texttt{:RND 6}.
Murmur closely tracks Uniform.
Prefuzz works well at small sizes.
Adaptive is a bit better still.}
\label{fig:cons-rnd-6-regret}
\end{figure}

\begin{figure}[H]
\begin{tikzpicture}
  \begin{axis}[ylabel=ns / put,
      xmin=1, xmode=log, log basis x={2},
      ymin=20, ymax=180, ymode=log, log basis y={2},
      legend pos=north west,
      legend style={nodes={scale=0.7, transform shape}},
      legend cell align={left},
      height=0.6*\linewidth,
      width=0.98\linewidth,
    ]
    \pgfplotstableread{data/cons-rnd-6-murmur.tbl}{\sorted}
    \addplot[murmurhash] table [x=nkeys, y=putns] {\sorted};
    \addlegendentry{Murmur};

    \pgfplotstableread{data/cons-rnd-6-prefuzz.tbl}{\sorted}
    \addplot[prefuzzhash] table [x=nkeys, y=putns] {\sorted};
    \addlegendentry{Prefuzz};

    \pgfplotstableread{data/cons-rnd-6-flat.tbl}{\sorted}
    \addplot[constanthash] table [x=nkeys, y=putns] {\sorted};
    \addlegendentry{Co+Pr};

    \pgfplotstableread{data/cons-rnd-6-flat-safe-small.tbl}{\sorted}
    \addplot[adaptivehash] table [x=nkeys, y=putns] {\sorted};
    \addlegendentry{\textbf{Adaptive}};
  \end{axis}
\end{tikzpicture}
\caption{PUT timings with \texttt{CONS} \texttt{:RND 6}. Prefuzz outperforms Murmur even at large sizes despite its higher regret because its collisions are between subsequent elements of the progression, which is friendly to the cache.}
\label{fig:cons-rnd-6-put}
\end{figure}

\begin{figure}[H]
\begin{tikzpicture}
  \begin{axis}[ylabel=ns / get,
      xmin=1, xmode=log, log basis x={2},
      ymin=6, ymax=180, ymode=log, log basis y={2},
      legend pos=north west,
      legend style={nodes={scale=0.7, transform shape}},
      legend cell align={left},
      height=0.6*\linewidth,
      width=0.98\linewidth,
    ]
    \pgfplotstableread{data/cons-rnd-6-murmur.tbl}{\sorted}
    \addplot[murmurhash] table [x=nkeys, y=getns] {\sorted};
    \addlegendentry{Murmur};

    \pgfplotstableread{data/cons-rnd-6-prefuzz.tbl}{\sorted}
    \addplot[prefuzzhash] table [x=nkeys, y=getns] {\sorted};
    \addlegendentry{Prefuzz};

    \pgfplotstableread{data/cons-rnd-6-flat.tbl}{\sorted}
    \addplot[constanthash] table [x=nkeys, y=getns] {\sorted};
    \addlegendentry{Co+Pr};

    \pgfplotstableread{data/cons-rnd-6-flat-safe-small.tbl}{\sorted}
    \addplot[adaptivehash] table [x=nkeys, y=getns] {\sorted};
    \addlegendentry{\textbf{Adaptive}};
  \end{axis}
\end{tikzpicture}
\caption{GET timings with \texttt{CONS} \texttt{:RND 6}}
\label{fig:cons-rnd-6-get}
\end{figure}

\begin{figure}[H]
\begin{tikzpicture}
  \begin{axis}[ylabel=ns / miss,
      xmin=1, xmode=log, log basis x={2},
      ymin=6, ymax=180, ymode=log, log basis y={2},
      legend pos=north west,
      legend style={nodes={scale=0.7, transform shape}},
      legend cell align={left},
      height=0.6*\linewidth,
      width=0.98\linewidth,
    ]
    \pgfplotstableread{data/cons-rnd-6-murmur.tbl}{\sorted}
    \addplot[murmurhash] table [x=nkeys, y=missns] {\sorted};
    \addlegendentry{Murmur};

    \pgfplotstableread{data/cons-rnd-6-prefuzz.tbl}{\sorted}
    \addplot[prefuzzhash] table [x=nkeys, y=missns] {\sorted};
    \addlegendentry{Prefuzz};

    \pgfplotstableread{data/cons-rnd-6-flat.tbl}{\sorted}
    \addplot[constanthash] table [x=nkeys, y=missns] {\sorted};
    \addlegendentry{Co+Pr};

    \pgfplotstableread{data/cons-rnd-6-flat-safe-small.tbl}{\sorted}
    \addplot[adaptivehash] table [x=nkeys, y=missns] {\sorted};
    \addlegendentry{\textbf{Adaptive}};
  \end{axis}
\end{tikzpicture}
\caption{MISS timings with \texttt{CONS} \texttt{:RND 6}}
\label{fig:cons-rnd-6-miss}
\end{figure}

\begin{figure}[H]
\begin{tikzpicture}
  \begin{axis}[ylabel=ns / del,
      xmin=1, xmode=log, log basis x={2},
      ymin=6, ymax=180, ymode=log, log basis y={2},
      legend pos=north west,
      legend style={nodes={scale=0.7, transform shape}},
      legend cell align={left},
      height=0.6*\linewidth,
      width=0.98\linewidth,
    ]
    \pgfplotstableread{data/cons-rnd-6-murmur.tbl}{\sorted}
    \addplot[murmurhash] table [x=nkeys, y=delns] {\sorted};
    \addlegendentry{Murmur};

    \pgfplotstableread{data/cons-rnd-6-prefuzz.tbl}{\sorted}
    \addplot[prefuzzhash] table [x=nkeys, y=delns] {\sorted};
    \addlegendentry{Prefuzz};

    \pgfplotstableread{data/cons-rnd-6-flat.tbl}{\sorted}
    \addplot[constanthash] table [x=nkeys, y=delns] {\sorted};
    \addlegendentry{Co+Pr};

    \pgfplotstableread{data/cons-rnd-6-flat-safe-small.tbl}{\sorted}
    \addplot[adaptivehash] table [x=nkeys, y=delns] {\sorted};
    \addlegendentry{\textbf{Adaptive}};
  \end{axis}
\end{tikzpicture}
\caption{DEL timings with \texttt{CONS} \texttt{:RND 6}}
\label{fig:cons-rnd-6-del}
\end{figure}


\clearpage
\section{Results for \texttt{SYMBOL} \texttt{:EXISTING}}
\label{sec:results-for-symbol-existing}

\balance

We list all symbols in the running Lisp system and take a random subset of the desired size.
These symbols happen to be packed tightly in not too many pages, and their addresses, if sorted, would approximately follow an arithmetic progression, which is a great fit for Pointer-Shift (\Cref{sec:the-pointer-shift-hash}, \Cref{alg:rehash-eq}).
But the keys are randomly selected and presented in random order even for PUT, so this effect can only be seen at close to maximal sizes.

The MISS key set is generated with \lisp{(gensym)}.

\begin{figure}[H]
\begin{tikzpicture}
  \begin{axis}[xlabel=number of keys, ylabel=regret,
      xmin=1, xmode=log, log basis x={2},
      ymin=-0.01, ymax=0.79,
      legend pos=north west,
      legend columns=2,
      legend style={nodes={scale=0.7, transform shape}},
      legend cell align={left},
      height=0.6*\linewidth,
      width=0.98\linewidth,
    ]
    \tikzpicturedependsonfile{data/symbol-existing-murmur}
    \tikzpicturedependsonfile{data/symbol-existing-prefuzz}
    \tikzpicturedependsonfile{data/symbol-existing-flat}
    \tikzpicturedependsonfile{data/symbol-existing-flat-safe-small}

    \pgfplotstableread{data/symbol-existing-murmur.tbl}{\sorted}
    \addplot[uniformhash] table [x=nkeys, y=rndregret] {\sorted};
    \addlegendentry{Uniform};

    \pgfplotstableread{data/symbol-existing-murmur.tbl}{\sorted}
    \addplot[murmurhash] table [x=nkeys, y=regret] {\sorted};
    \addlegendentry{Murmur};

    \pgfplotstableread{data/symbol-existing-prefuzz.tbl}{\sorted}
    \addplot[prefuzzhash] table [x=nkeys, y=regret] {\sorted};
    \addlegendentry{Prefuzz};

    \pgfplotstableread{data/symbol-existing-flat.tbl}{\sorted}
    \addplot[constanthash] table [x=nkeys, y=regret] {\sorted};
    \addlegendentry{Co+Pr};

    \pgfplotstableread{data/symbol-existing-flat-safe-small.tbl}{\sorted}
    \addplot[adaptivehash] table [x=nkeys, y=regret] {\sorted};
    \addlegendentry{\textbf{Adaptive}};
    \addplot[mark=*,mark options={scale=1.0}] coordinates {(33,0)};
  \end{axis}
\end{tikzpicture}
\caption{Regret with \texttt{SYMBOL} \texttt{:EXISTING}.
All hashes closely track Uniform.
Adaptive, which leaves memory addresses most intact, takes advantage of the nature of the data at close to maximal sizes.}
\label{fig:symbol-existing-regret}
\end{figure}

\begin{figure}[H]
\begin{tikzpicture}
  \begin{axis}[ylabel=ns / put,
      xmin=1, xmode=log, log basis x={2},
      ymin=20, ymax=180, ymode=log, log basis y={2},
      legend pos=north west,
      legend columns=2,
      legend style={nodes={scale=0.7, transform shape}},
      legend cell align={left},
      height=0.6*\linewidth,
      width=0.98\linewidth,
    ]
    \tikzpicturedependsonfile{data/symbol-existing-murmur}
    \tikzpicturedependsonfile{data/symbol-existing-prefuzz}
    \tikzpicturedependsonfile{data/symbol-existing-flat}
    \tikzpicturedependsonfile{data/symbol-existing-flat-safe-small}

    \pgfplotstableread{data/symbol-existing-murmur.tbl}{\sorted}
    \addplot[murmurhash] table [x=nkeys, y=putns] {\sorted};
    \addlegendentry{Murmur};

    \pgfplotstableread{data/symbol-existing-prefuzz.tbl}{\sorted}
    \addplot[prefuzzhash] table [x=nkeys, y=putns] {\sorted};
    \addlegendentry{Prefuzz};

    \pgfplotstableread{data/symbol-existing-flat.tbl}{\sorted}
    \addplot[constanthash] table [x=nkeys, y=putns] {\sorted};
    \addlegendentry{Co+Pr};

    \pgfplotstableread{data/symbol-existing-flat-safe-small.tbl}{\sorted}
    \addplot[adaptivehash] table [x=nkeys, y=putns] {\sorted};
    \addlegendentry{\textbf{Adaptive}};
  \end{axis}
\end{tikzpicture}
\caption{PUT timings with \texttt{SYMBOL} \texttt{:EXISTING}.
Prefuzz outperforms Murmur because it is faster to compute.
The advantage of Co+Pr and Adaptive over Prefuzz comes almost exclusively from the initial Constant hash phase, with Adaptive enjoying a small edge due to its better regret at the very end.}
\label{fig:symbol-existing-put}
\end{figure}

\begin{figure}[H]
\begin{tikzpicture}
  \begin{axis}[ylabel=ns / get,
      xmin=1, xmode=log, log basis x={2},
      ymin=6, ymax=180, ymode=log, log basis y={2},
      legend pos=north west,
      legend style={nodes={scale=0.7, transform shape}},
      legend cell align={left},
      height=0.6*\linewidth,
      width=0.98\linewidth,
    ]
    \tikzpicturedependsonfile{data/symbol-existing-murmur}
    \tikzpicturedependsonfile{data/symbol-existing-prefuzz}
    \tikzpicturedependsonfile{data/symbol-existing-flat}
    \tikzpicturedependsonfile{data/symbol-existing-flat-safe-small}

    \pgfplotstableread{data/symbol-existing-murmur.tbl}{\sorted}
    \addplot[murmurhash] table [x=nkeys, y=getns] {\sorted};
    \addlegendentry{Murmur};

    \pgfplotstableread{data/symbol-existing-prefuzz.tbl}{\sorted}
    \addplot[prefuzzhash] table [x=nkeys, y=getns] {\sorted};
    \addlegendentry{Prefuzz};

    \pgfplotstableread{data/symbol-existing-flat.tbl}{\sorted}
    \addplot[constanthash] table [x=nkeys, y=getns] {\sorted};
    \addlegendentry{Co+Pr};

    \pgfplotstableread{data/symbol-existing-flat-safe-small.tbl}{\sorted}
    \addplot[adaptivehash] table [x=nkeys, y=getns] {\sorted};
    \addlegendentry{\textbf{Adaptive}};
  \end{axis}
\end{tikzpicture}
\caption{GET timings with \texttt{SYMBOL} \texttt{:EXISTING}.
}
\label{fig:symbol-existing-get}
\end{figure}

\begin{figure}[H]
\begin{tikzpicture}
  \begin{axis}[ylabel=ns / miss,
      xmin=1, xmode=log, log basis x={2},
      ymin=6, ymax=180, ymode=log, log basis y={2},
      legend pos=north west,
      legend style={nodes={scale=0.7, transform shape}},
      legend cell align={left},
      height=0.6*\linewidth,
      width=0.98\linewidth,
    ]
    \tikzpicturedependsonfile{data/symbol-existing-murmur}
    \tikzpicturedependsonfile{data/symbol-existing-prefuzz}
    \tikzpicturedependsonfile{data/symbol-existing-flat}
    \tikzpicturedependsonfile{data/symbol-existing-flat-safe-small}

    \pgfplotstableread{data/symbol-existing-murmur.tbl}{\sorted}
    \addplot[murmurhash] table [x=nkeys, y=missns] {\sorted};
    \addlegendentry{Murmur};

    \pgfplotstableread{data/symbol-existing-prefuzz.tbl}{\sorted}
    \addplot[prefuzzhash] table [x=nkeys, y=missns] {\sorted};
    \addlegendentry{Prefuzz};

    \pgfplotstableread{data/symbol-existing-flat.tbl}{\sorted}
    \addplot[constanthash] table [x=nkeys, y=missns] {\sorted};
    \addlegendentry{Co+Pr};

    \pgfplotstableread{data/symbol-existing-flat-safe-small.tbl}{\sorted}
    \addplot[adaptivehash] table [x=nkeys, y=missns] {\sorted};
    \addlegendentry{\textbf{Adaptive}};
  \end{axis}
\end{tikzpicture}
\caption{MISS timings with \texttt{SYMBOL} \texttt{:EXISTING}}
\label{fig:symbol-existing-miss}
\end{figure}

\begin{figure}[H]
\begin{tikzpicture}
  \begin{axis}[ylabel=ns / del,
      xmin=1, xmode=log, log basis x={2},
      ymin=6, ymax=180, ymode=log, log basis y={2},
      legend pos=north west,
      legend style={nodes={scale=0.7, transform shape}},
      legend cell align={left},
      height=0.6*\linewidth,
      width=0.98\linewidth,
    ]
    \tikzpicturedependsonfile{data/symbol-existing-murmur}
    \tikzpicturedependsonfile{data/symbol-existing-prefuzz}
    \tikzpicturedependsonfile{data/symbol-existing-flat}
    \tikzpicturedependsonfile{data/symbol-existing-flat-safe-small}

    \pgfplotstableread{data/symbol-existing-murmur.tbl}{\sorted}
    \addplot[murmurhash] table [x=nkeys, y=delns] {\sorted};
    \addlegendentry{Murmur};

    \pgfplotstableread{data/symbol-existing-prefuzz.tbl}{\sorted}
    \addplot[prefuzzhash] table [x=nkeys, y=delns] {\sorted};
    \addlegendentry{Prefuzz};

    \pgfplotstableread{data/symbol-existing-flat.tbl}{\sorted}
    \addplot[constanthash] table [x=nkeys, y=delns] {\sorted};
    \addlegendentry{Co+Pr};

    \pgfplotstableread{data/symbol-existing-flat-safe-small.tbl}{\sorted}
    \addplot[adaptivehash] table [x=nkeys, y=delns] {\sorted};
    \addlegendentry{\textbf{Adaptive}};
  \end{axis}
\end{tikzpicture}
\caption{DEL timings with \texttt{SYMBOL} \texttt{:EXISTING}}
\label{fig:symbol-existing-del}
\end{figure}

\end{document}